\documentclass[submission,copyright,creativecommons]{eptcs}
\providecommand{\event}{NCMA 2024}

\usepackage{iftex}

\ifpdf
  \usepackage{underscore}         
  \usepackage[T1]{fontenc}        
\else
  \usepackage{breakurl}           
\fi

\usepackage[utf8]{inputenc}
\usepackage{url,hyperref}
\usepackage{array}
\usepackage{amsmath,amstext,amssymb,stmaryrd}
\usepackage{mathtools}
\usepackage{bbold}

\usepackage[ruled,vlined,linesnumbered]{algorithm2e}

\usepackage{xcolor}
\usepackage{latexcolors}
\usepackage{graphicx}

\usepackage{latexsym}

\usepackage{tikz}

\usetikzlibrary{graphs}
\usetikzlibrary{positioning}
\usetikzlibrary{arrows.meta,automata,positioning}
\usetikzlibrary{decorations.pathmorphing}
\usetikzlibrary{petri}
\usetikzlibrary{shapes}
\usetikzlibrary{calc}
\usetikzlibrary{shadows}
\usetikzlibrary{patterns, patterns.meta}

\usepackage[]{proof-at-the-end}

\tikzstyle{place}=[circle,draw,minimum size=2mm]
\tikzstyle{transition}=[rectangle,draw,fill=black,minimum width=6mm,inner ysep=1pt]
\tikzset{
	every token/.style={minimum size=1pt},
	token distance=3pt
	anchor=base,
}
\tikzset{
    vertex/.style={
      draw, 
      circle,
      fill,
      inner sep=0pt,
      outer sep=0pt,
      minimum size=0pt,
    }
}

\newcommand\tbd[1]{#1}
\newcommand{\minisec}[1]{} 

\newcommand\formalism{\gram}
\newcommand\dagformalism{DAG grammar $\gram = (N,\Sigma,R)$}

\let\mod\relax\DeclareMathOperator{\mod}{{}\mathrm{mod}{}}

\newcommand{\mul}[1]{\mathcal M(#1)}

\newcommand\nullconfig{\mathbf{0}} 
\newcommand{\meta}[1]{\underline{#1}}
\newcommand{\metaQ}{\mathcal Q}

\newcommand{\metaQtext}[1]{\metaQ{}_{\text{#1}}}
\newcommand{\metaq}{\mathbb{q}}

\newcommand\metaQmin{\metaQtext{min}}
\newcommand\metaqmin{\metaQtext{min}(\gram)}

\newcommand\metaQzero{\metaQ_\nullconfig} 
\newcommand\metaqzero{\metaQzero(\gram)}

\newcommand{\gram}{\mathcal G}
\newcommand\gramdef{\gram = (N, \Sigma, R)}
\newcommand\dagdef{G = (V, E, \lab, \IN, \OUT)}
\newcommand{\To}{\Rightarrow}
\newcommand{\ToR}{\To^*_R}

\newcommand{\Tor}{\To_r}
\newcommand{\ToPi}{\To_\Pi}

\newcommand\Gstar{\gram_{star}}
\newcommand\Gtree{\gram_{tree}}
\newcommand\Groot{G_\text{\textit{root}}}

\newcommand\rr{r} 
\newcommand\rl{l} 

\newcommand{\anchor}[2]{\tikz[remember picture,baseline=-.3ex,inner xsep=0,inner ysep=2pt]{\node (#1) {$#2$};}}
\usetikzlibrary{calc}
\newcommand*\circled[1]{%
  \tikz[%
        node distance=7pt,
        baseline=(vertexlabel.base),
        edgestring/.style={inner sep=0pt},
        vertex/.style={
                       draw,
                       shape=rounded rectangle,
                       inner ysep=0pt,
                       inner xsep=-2pt
                      }
       ]{
    \node[vertex] (vertexlabel) {\rule[-3pt]{0pt}{\dimexpr2ex+2pt}#1};
    \node[edgestring, left=of vertexlabel]  (ingoing)  {};
    \node[edgestring, right=of vertexlabel] (outgoing) {};
    \draw [->>] (ingoing) -- (vertexlabel);
    \draw [->>] (vertexlabel) -- (outgoing);
}}
\newcommand{\lab}{\mathit{\ell}}
\newcommand{\IN}{\mathit{in}}
\newcommand{\OUT}{\mathit{out}}
\newcommand{\src}{\mathit{src}}
\newcommand{\tar}{\mathit{tar}}
\newcommand{\D}{\mathcal D}

\newcommand{\labelr}{\mathtt{r}}
\newcommand{\labelc}{\mathtt{c}}
\newcommand{\labell}{\mathtt{l}}
\newcommand{\labelm}{\mathtt{m}}
\newcommand{\labelu}{\textup{u}} 
\newcommand{\rarrow}[1]{\circled{$\mathtt{#1}$}}
\newcommand{\rarrowl}{\rarrow{l}}
\newcommand{\sigmaarrow}{\circled{$\sigma$}}
\newcommand{\sigmaArrow}[1]{\circled{$\sigma_{#1}$}}
\newcommand{\earrow}{\circled{}} 
\newcommand{\union}{\mathbin{\&}}

\newcommand{\swapdee}{D(e_0 \bowtie e_1)}
\newcommand{\gswap}{\DAG\swapdee}
\newcommand{\swapping}[4]{#1(#2\bowtie#3)^{#4}}
\newcommand{\DAG}[1]{\lfloor #1\rfloor}
\newcommand\gd{\DAG D}

\newcommand{\RDL}{\ensuremath{\mathit{RDL}}}
\newcommand{\RDLdet}{\ensuremath{\mathit{RDL^\text{det}}}}

\newcommand{\emptygraph}{\varnothing}

\newcommand\pcolor{tyrianpurple}
\newcommand\qcolor{pakistangreen}
\newcommand\colorp{\color{\pcolor}} 
\newcommand\colorq{\color{\qcolor}} 
\newcommand\p{{\colorp p}}
\newcommand\q{{\colorq q}}

\newcommand\starrone{$\lambda \rarrow{r} qp$}
\newcommand\bowrzero{$\lambda \earrow{} pq$}
\newcommand\bowrone{$p \earrow{} pq$}
\newcommand\starrtwo{$pq \rarrow{l} \lambda $}
\newcommand\bowrtwo{$pq \earrow{} p $}
\newcommand\bowrtwozero{$pq \earrow{} \lambda $}

\newcommand\rsymbol{\mathtt{r}}
\newcommand\msymbol{\mathtt{m}}
\newcommand\lsymbol{\mathtt{L}}
\newcommand\osymbol{\mathtt{o}}

\newcommand\therule[1]{r_{#1}}

\newcommand\mrule{\therule{m}}
\newcommand\lrule{\therule{l}}

\renewcommand\mul[1]{\mathbb N^{#1}}
\newcommand\groot{G_{\text{root}}}
\newcommand\metagroot{\meta G_{\text{root}}}

\newcommand\vnext{v_\text{next}}

\newcommand\vroot{v_\text{root}}
\newcommand\Lgram{L(\gram)}
\newcommand\Lall{L(\gram)^{\&}} 

\newcommand\Gbow{\gram_\text{bow}}
\newcommand\Lbow{L(\Gbow)}

\newcommand\LQ{L^\metaQ(\formalism)}
\newcommand\LQbow{L^\metaQ(\Gbow)}

\newcommand\m[1]{\bar #1}  
\newcommand\mex[1]{\hat #1} 
\newcommand\men[1]{\check #1} 

\newcommand\FID{\mathtt{FID}} 

\newcommand\FD{\mathtt{FD}} 

\newcommand\ID{\mathtt{ID}}

\newcommand\AMR{\texttt{AMR}}

\newcommand\lipicsfont[1]{{\textcolor{lipicsGray}{\sffamily\bfseries\upshape\mathversion{bold}(#1)}}}
\newcommand\lipicsrightdescr[1]{\hfill\lipicsfont{#1}}

\newcommand\mgrule{(\m\alpha\sigmaarrow\m\beta)}
\newcommand\rmgrule{\m r = \mgrule}

\newcommand\refstandard{\cite{journals/iandc/BlumDrewes2019}}

\usepackage{subcaption} 
\usepackage{amsthm}     
\newtheorem{theorem}{Theorem}[section]
\newtheorem{lemma}[theorem]{Lemma}

\newtheorem{corollary}[theorem]{Corollary}
\newtheorem{observation}[theorem]{Observation}

\newtheorem{definition}[theorem]{Definition}
\newtheorem{example}[theorem]{Example}

\newtheorem*{pdconjecture*}{Pushdown Conjecture}

\title{A New Notion of Regularity:\\ Finite State Automata Accepting Graphs}
\author{Yvo Ad Meeres
\institute{Department of Theoretical Computer Science\\
University of Bremen\\
Bremen, Germany}
\email{yvo.meeres@mailbox.org}
}

\newcommand{\titlerunning}{Finite State Automata Accepting Graphs}
\newcommand{\authorrunning}{Yvo Ad Meeres}

\hypersetup{
  bookmarksnumbered,
  pdftitle    = {\titlerunning},
  pdfauthor   = {\authorrunning},
  pdfsubject  = {EPTCS},               
}
 \definecolor[named]{lipicsGray}{rgb}{0.31,0.31,0.33}

\begin{document}
\maketitle

\begin{abstract}
Analogous to regular string and tree languages,
 regular languages of directed acyclic graphs (DAGs) are defined in the literature.
Although called regular, those DAG-languages are more powerful and, consequently,
standard problems have a higher complexity than in the string case.
Top-down as well as bottom-up deterministic DAG languages are subclasses of the regular DAG languages.
We refine this hierarchy by providing a weaker subclass of the deterministic DAG languages.
For a DAG grammar generating a language in this new DAG language class, or,
equivalently, a DAG-automaton recognizing it,
a classical deterministic finite state automaton (DFA) can be constructed.
As the main result, we provide a characterization of this class.

The motivation behind this is the transfer of techniques for regular string languages to graphs. 
Trivially, our restricted DAG language class is closed under union and intersection.
This permits the application of minimization and hyper-minimization algorithms known for DFAs.
This alternative notion of regularity coins at the existence of a DFA for recognizing a DAG language.
\end{abstract}


\section{Introduction}

\begin{figure}[t]
\includegraphics[width=1\textwidth]{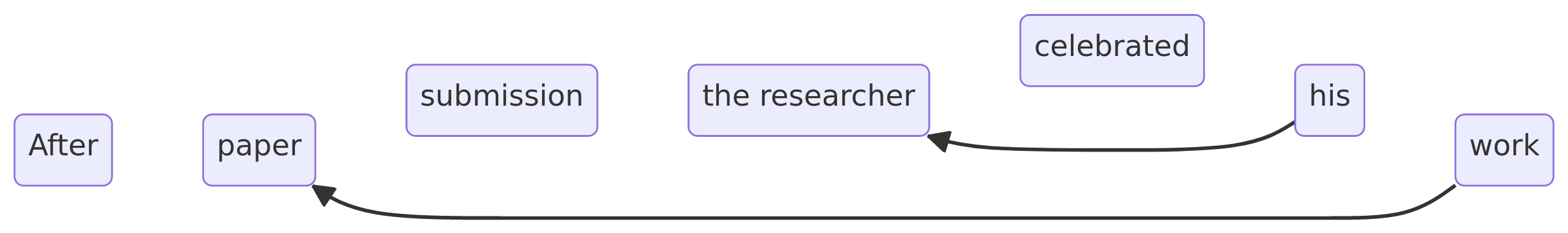}
\caption{Classical NLP parsing is blind to coreferences within sentences,
since trees cannot represent these edges within the parse tree.
Graphs, on the contrary, are capable of showing coreferences between e.g.~words of parsed sentences.
For the above sentence,
its parse trees could neither model the obvious possessive relation
between the possessive pronoun \emph{his} and \emph{the researcher}
nor the semantic kind of equivalence relation requiring world knowledge
between the \emph{paper} and the \emph{work}.
But, a semantic graph like an \AMR{} DAG~\cite{quernheim-knight:12b} could.
The capabilities of semantic graphs
are illustrated in~\cite{daggrammar} 
as well as for a complex sentence in~\cite{DBLP:conf/mol/Drewes17}
by means of the representation of a sentence as an \AMR{} DAG. 
}\label{fig/intro-nlp-dag}
\end{figure}

Many research fields either struggle with the complexity of processing graphs
-- for example fields like high performance computing~\cite{DBLP:journals/corr/GharaibehSCR13}
or neurocomputing~\cite{DBLP:journals/ijon/BoccatoFDT24}, to mention just a few --
or by encoding their graph problems as strings, see e.g.~\cite{DBLP:journals/pr/DuttaLP13, DBLP:journals/spe/DiasPDA14}.
The well-researched class of regular string languages,
recognized by finite state automata (FSAs),
exhibits a fruitful balance between
expressiveness and efficiency concerning standard algorithmic problems.
The problem is, that these algorithms are only applicable to strings.
One approach would be to provide efficient graph algorithms
for specific problems,
as in~\cite{DBLP:conf/atmos/MatsubayashiS23} for the membership problem
circumventing Braess's Paradox \cite{DBLP:journals/mmor/Braess68}
or in~\cite{DBLP:journals/jgaa/BestW22} providing a faster algorithm
for the very specific problem
of the maximum independent set on interval filament graphs.
Instead of the cumbersome approach to tackle all these specific problems one by one, 
our meta-approach suggests porting all efficient algorithms known
for string processing to graph processing in one sweep.
To port a wide range of well-known efficient algorithms (based on FSAs)
from strings to graphs,
this article introduces FSAs recognizing sets of
directed acyclic graphs (DAGs) instead of sets of strings.
Such sets are called DAG languages.
We consider \emph{vertex-labeled DAGs} with unlabeled edges
but label those edges for accepting a DAG.
A classical FSA accepts a string by reading it symbol by symbol.
Our proposed FSA accepts a DAG by reading top-down vertex by vertex instead.
A symbol read by the automaton encodes a whole vertex,
consisting of its vertex label and 
its ordered and labeled in- and outgoing edges.
While reading the vertices top-down,
the outgoing edges in a DAG are labeled according to the automaton’s specification
while the ingoing edges, labeled beforehand, have to match.
The FSA's states keep track of those ingoing edge labels
whose target vertices are not yet read,
thus whose outgoing edges are unlabeled.
The class of DAG languages accepted
by such an FSA is a proper subset of the top-down deterministic regular DAG languages 
defined in the literature~\cite{journals/iandc/BlumDrewes2019}.
This class, in turn, is a proper subset of
the regular DAG languages~\cite{journals/iandc/BlumDrewes2019}.

In literature, the notion of regularity concerning DAG languages differs
from that applied to string languages. Regular DAG automata recognize regular DAG languages~\cite{journals/iandc/BlumDrewes2019}.
These automata are one of the formalisms~\cite{DBLP:conf/mol/Drewes17} proposed in the literature
to model semantics by using Abstract Meaning Representation~(\AMR{})~\cite{banarescu.etal:13}.
DAG automata were originally introduced by Kamimura and Slutzki \cite{kamimura-slutzki:81,kamimura-slutzki:82}.
A promising alternative formalism, not considered in this paper, is the hyperedge replacement graph grammar \cite{DBLP:conf/gg/HabelK86}.
Classical natural language parsing turns a sentence into a parse tree while semantic parsing,
such as \AMR{} parsers~\cite{damonte2018cross-lingual,
damonte2017incrememntal, DBLP:journals/linguamatica/SenoCIAR22},
can model coreferences between e.g. the words of a sentence.
Fig.~\ref{fig/intro-nlp-dag} shows two such relations which turn a parse tree into a DAG.
Such semantic relations, expressible in \AMR{},
specify that the work was conducted by the mentioned researcher (\emph{researcher} $\leftarrow$ \emph{his})
and writing the paper is the researcher's daily work (\emph{work} $\rightarrow$ \emph{paper}).
Semantic parsing provides vital contributions to improve natural language processing~(NLP).
The anecdotes about AI chat bots inventing facts feed wishes for NLP improvements by ensuring semantic consistency.
This consistency is desirable also for sophisticated spell and grammar checking and machine translations.

The membership problem for regular DAG automata
surprisingly being NP-complete~\cite{DBLP:journals/coling/ChiangDGLS18,DBLP:journals/tcs/Fujiyoshi10},
the uniform and even the non-uniform one\footnote{%
The uniform membership problem asks for a given automaton and a given graph whether the graph is an element of the given language;
the non-uniform membership problem asks for a fixed automaton and a given graph whether the automaton accepts the graph, making the membership problem potentially easier.%
},
finding strategies for identifying efficient semantic parsing algorithms remains the core problem.
By determinizing, either top-down or bottom-up, the membership problem becomes tractable.
However, even with restrictions 
like determinism or planarity~\cite{kamimura-slutzki:81,DBLP:journals/corr/abs-1810-12266},
parsing problems can easily become too complex.
For instance, Vasiljeva et al. were surprised that for certain probabilistic
DAG automata non-trivial probability distributions are necessary
to assign weights \cite{DBLP:journals/corr/abs-1810-12266}.
Since the notions of regularity differ for the string and DAG case, we propose the \emph{new notion of regularity}: use the string case regularity for DAGs in order to obtain better algorithmic properties. The regularity notion for DAGs, presented in literature, seems to be too powerful to provide efficient algorithms. Viewing DAG languages only then as regular when they can be recognized by an FSA, provides deep insight into the structural properties of DAG languages.

Although the mildly context sensitive upper bound for natural language parsing
classifies parsing as lying between context-free and context-sensitive formalisms,
finite state descriptions of languages are of major importance \cite{hurskainen2023rule}.
For the sake of efficiency, this field often seeks to digest also mildly context sensitive structures
with finite state methods \cite{pirinen2023finite,linden2009weighting, pirinen2010finite},
an approach conceivable also for DAG digestion.
The conference \emph{Finite State Methods in Natural Language Processing} (FSMNLP)
concentrates on this lowest level of the Chomsky hierarchy.
Many NLP tools, like \texttt{apertium}, \texttt{HFST} or \texttt{GiellaLT}~\cite{hfst2023, swanson2023apertium, moshagen2023giellalt}
operate with finite state descriptions.

~


\begin{figure}[t]
\centering
\subfloat[Classical and DAG regularity]{
  \label{pic:regularity}
  \begin{tikzpicture}[semic/.style args={#1,#2}{semicircle,minimum  
width=#1,draw,anchor=arc end,rotate=#2},outer sep=0pt,line width=.7pt]

\node[align=center] at (3,-3.2){\emph{classical}\\\emph{regularity:}\\$\FD$};
\draw(1.55,-3) -- (2.1,-3.2);

\draw[pattern = {Dots}] (0.5, -3) circle (1cm) {};
\draw[pattern = {Dots}] (-0.5, -3) circle (1cm) {};
\fill[fill = white] (-0.5, -2) rectangle (0.5, -4);
\fill[pattern = {Dots}] (-0.5, -2) rectangle (0.5, -4);
\draw[clip] (0,0) circle (3cm);
\draw[black, thick, align=center] (0,0) circle (3cm) node[above=0cm] {\emph{DAG}\\\emph{regularity:}\\$\ID{} \cup \FID$};
\node[align=center] at (1,-1){$\ID$};
\node[align=center] at (0,-2.5){$\FID$};
\end{tikzpicture}
}
\hfill
\subfloat[DAG language classes]{
  \label{pic:dag-classes}
  \begin{tikzpicture}[semic/.style args={#1,#2}{semicircle,minimum  
width=#1,draw,anchor=arc end,rotate=#2},outer sep=0pt,line width=.7pt]

\node at (2,-3){$\FD$};
\draw(1.55,-3) -- (1.8,-3);
\node at (2.5,-2.5){$\FID$};
\draw(1,-2.5) -- (2.2,-2.5);
\node (id) at (3.1,-1.5){$\ID$};
\draw[black] (id) -- (2.2,-.5);
\draw (id) -- (2,-1.5);

\fill[green, fill opacity=0.3] (-0.5, -1) rectangle (0.5, -2);
\begin{scope}
\draw[clip] (0,0) circle (3cm);
\fill[yellow, fill opacity = 0.3] (-3, -1) rectangle (-0.5, -4);
\fill[blue, fill opacity = 0.3] (0.5, -1) rectangle (3, -4);
\end{scope}
\begin{scope}
\fill[green, fill opacity = 0.5] (-0.5, -2) rectangle (0.5, -4);
\draw[pattern = {Dots}] (-0.5, -2) rectangle (0.5, -4);
\end{scope}
\draw[black] (-0.5, -2) -- (0.5, -2);
\draw[black] (-0.5, -1) rectangle (0.5, -4);
\node[fill = white, semic={2cm, 90}] at (-0.5, -4) {};
\node[fill = yellow, fill opacity = 0.7, semic={2cm, 90}] at (-0.5, -4) {};
\node[pattern = {Dots}, semic={2cm, 90}] at (-0.5, -4) {};
\node[fill = white, semic={2cm, 270}] at (0.5, -2) {};
\node[fill = blue, fill opacity = 0.5, semic={2cm, 270}] at (0.5, -2) {};
\node[pattern = {Dots}, semic={2cm, 270}] at (0.5, -2) {};
\draw[clip] (0,0) circle (3cm);
\begin{scope}
\draw[clip] (0,0) circle (3cm);
\node[fill = white, semic={2cm, 90}] at (-0.5, -4) {};
\node[fill = yellow, fill opacity = 0.3, semic={2cm, 90}] at (-0.5, -4) {};
\node[pattern = {Dots}, semic={2cm, 90}] at (-0.5, -4) {};
\node[fill = white, semic={2cm, 270}] at (0.5, -2) {};
\node[fill = blue, fill opacity = 0.3, semic={2cm, 270}] at (0.5, -2) {};
\node[pattern = {Dots}, semic={2cm, 270}] at (0.5, -2) {};
\fill[white] (-0.5, -2) rectangle (0.5, -4);
\fill[green, fill opacity = 0.3] (-0.5, -2) rectangle (0.5, -4);
\draw[pattern = {Dots}] (-0.5, -2) rectangle (0.5, -4);
\end{scope}
\draw[black, thick] (0,0) circle (3cm) node[above=0cm] {regular DAG languages};
\draw[black] (-3, -1) -- (3,-1);
\end{tikzpicture}
}
\captionsetup{subrefformat=parens}
\caption{%
\em
Overview over the language classes
\\
In both Venn diagrams,
the circle denotes the the regular DAG languages
whereas the oval denotes the language class $\FD$.
The intersection between the two is the class $\FID$
which is both closed under edge swap as well as under DFA-construction.
The dotted part, the oval, corresponds to $\FD$.
The non-dotted part corresponds to $\ID$.
\\
\subref{pic:regularity}
Classically, the term \emph{regularity} refers to FSAs
and thus to string languages.
This does not match the notion of regularity for DAG languages.
The two notions match only for languages in $\FID$.
\\
\subref{pic:dag-classes}
Top-down determinism and bottom-up determinism are colored in yellow and blue.
Consequently, green stands for languages which are both top-down as well as bottom-up deterministic.
In the right Venn diagram,
whereas all colored fields are deterministic,
the nondeterministic part corresponds to the white part
The class $\ID$ comprises those regular DAG languages which are not in $\FID$ (and consequently not in $\FD$),
and which are either (top-down / bottom-up) deterministic or non-deterministic.}
\label{fig:classes}
\end{figure}
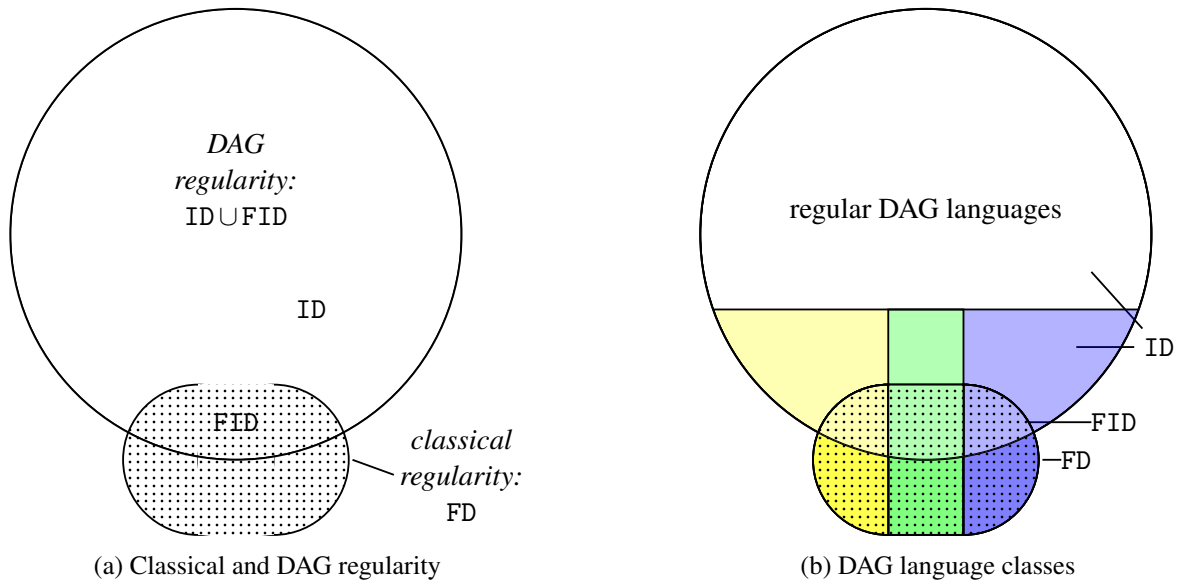

All this said, our contributions,
see  them illustrated in Fig.~\ref{fig:classes}, can be stated as
\begin{itemize}
\item the idea of using a classical FSA to recognize a DAG language
\item the notion of a \emph{meta-state}, a multiset of edge labels, serving as the states of the FSA
\item separation of the regular top-down deterministic DAG languages into those recognizable ($\FID$)
and those not recognizable ($\ID$) via an FSA by means of the meta-state technique
\item restricting a DAG automaton by meta-states, resulting in the class $\FD$ comprising $\FID$
      but not being a subset of the regular DAG languages
\item characterization of the newly defined classes (\emph{main result}).
\end{itemize}

Providing an FSA for a DAG language immediately opens up
a wide range of results for DAGs formerly applicable only to strings.
The folklore algorithm of FSA minimization can be applied,
just as algorithms for lossy FSA compression,
called hyper-minimization \cite{DBLP:journals/ijfcs/Badr09,DBLP:journals/ijfcs/MalettiQ11},
where \emph{hyper} stands for the tolerance of finitely many errors.
Morphological transducers already prove
hyper-minimizations being useful for NLP~\cite{pirinen2014hyper}.

Also from a structural point of view a deeper understanding of DAG language classes
and  a suitable overall hierarchy would be a beautiful result for theoretical computer science.
As mentioned in the beginning, 
many research fields require efficient graph algorithms
and therefore could potentially profit from these very limited DAG languages since they are parsable as efficiently as regular string languages. Even though their expressiveness is quite limited, implying that we cannot encode arbitrary graph languages, for a variety of important problems, the limited expressiveness will suffice, and algorithms can be ported directly from the string case.

\section{Preliminaries}\label{notation}

The set of non-negative integers is denoted by $\mathbb{N}$.
For $n \in \mathbb{N}$ we define $[n] = \{1,\dots,n\}$.
For a set $A$, we denote its cardinality by $|A|$.
A finite set $A$ is called an \emph{alphabet},
an element $a \in A$ is a \emph{symbol},
a \emph{string} is the concatenation of symbols,
the set of all finite strings over $A$ is denoted by $A^*$
and a, not necessarily proper, subset of $A^*$ is called a \emph{language}.
The empty string of length $0$ is denoted by $\lambda$.
The length of a string $w \in A^*$ is denoted by $|w|$ and
$[w]$ denotes the smallest set $A$ such that $w\in A^*$.
The concatenation of two strings $a$ and $b$ is written as its juxtaposition~$ab$.
For a string $w_1w_2 \dots w_i \dots w_n$ of length $n$ over $A$,
the \emph{position} of a symbol $w_i \in A$ is $i \in [n]$.
%
The canonical extensions of a function $f\colon X\to Y$
to the power set of $X$ and to $X^*$
are denoted by $f$ as well.
Thus, $f(\{x_1,\dots,x_n\})=\{f(x_1),\dots,f(x_n)\}$
and   $f(x_1\cdots x_n)    =f(x_1)\cdots f(x_n)$
for all $x_1,\dots,x_n\in X$.
For a set $Y$ a (locally finite) multiset over $Y$ is a function $\mu\colon Y\to\mathbb N$.
For brevity, we give a specific multiset 
by a string notation $y_{1}^{\mu(y_{1})}\cdots y_{n}^{\mu(y_{n})}$ for $y_1 \dots y_n \in Y$.
The size of $\mu$ is $|\mu|=\sum_{y\in Y}\mu(y)$. (Formally, $|\mu|=\infty$ if $\mu(y)\ge 1$ for infinitely many $y$, but this case will not be relevant for this paper, i.e., all multisets appearing here will be finite.)
We denote the set of all multisets over $Y$ by $\mul Y$.
For a function $f\colon X\to Y$, we let $f_M\colon 2^X\to\mul Y$ be the mapping such that,
for every $X'\subseteq X$, $f_M(X')$ is the multiset of images of elements of $X'$ under $f$. 
Thus, formally, $f_M(X')(y)=|\{x\in X'\mid f(x)=y\}|$ for every $y\in Y$. 

This article studies languages of vertex-labeled, directed multigraphs without loops and
with ordered unlabeled edges (called \emph{graphs}, see~Def.\,\ref{def:graph})
which are acyclic (called \emph{DAGs}, see~Def.\,\ref{def:dag}).
Edges will be labeled only temporarily
by a  grammar (c.f.~Def.\,\ref{def:grammar}), its equivalent automaton
or classical finite state automaton (see~Section\,\ref{sec:fsa}).
\begin{definition}[\tbd{Graph}]\label{def:graph}
A \emph{graph} over $\,\Gamma$ is a tuple $\dagdef$
with $\,\Gamma$, $V\!$ and $E$ being disjoint finite sets,
the sets of \emph{vertex labels}, \emph{vertices} and \emph{edges}, respectively.
The vertices are labeled by $\lab \colon V \rightarrow \Gamma$.
For an \emph{edge} $e \in E$ 
between
the vertices $(v, w) \in V \times V\!$,
\emph{directed} from $v$ to $w$, with $v \neq w$,
the \emph{source} $v$ is referenced by $\src(e)$ and 
the \emph{target} $w$ by $\tar(e)$.
By $\IN, \OUT \colon V \rightarrow E^*$ we assign to each vertex $v \in V$
its \emph{incoming} and \emph{outgoing} edges
such that
$\src(e) = v \Leftrightarrow e \in [\OUT(v)]$
and
$\tar(e) = v \Leftrightarrow e \in [\IN(v)]$.
These edges are ordered as specified
by the strings $\IN(v)$ and $\OUT(v)$.
%
The \emph{empty graph} $\emptygraph$ is the graph whose set of vertices is empty.
A vertex is called a \emph{root} or a \emph{leaf}
if $\IN(v)$ or $\OUT(v)$ are empty, respectively.
The disjoint union of graphs, meaning disjoint sets of vertices and edges,
is denoted by the operator $\&$.
\end{definition}
%
\begin{definition}[\tbd{Path}]\label{def:path}%
A \emph{path} in a graph $\dagdef$
is a nonempty sequence of edges $e_1,\dots,e_n$,
$e_i \in E$ for $i \in [n]$,
yielding a unique alternating sequence $v_0e_1v_1 \cdots e_nv_n$
with vertices $v_0,\dots,v_n\in V\!$ such that
$\{\src(e_i),$ $\tar(e_i)\}= \{v_{i-1}, v_i\}$ for all $i\in[n]$.
Such a path is a \emph{cycle} if $v_0=v_n$.
A \emph{path between} $s$ and $t$ is a path with
$s \in \{v_0, e_1\}$ and
$t \in \{v_n, e_n\}$.
A path is \emph{directed} if for $i \in [n]$ either
$\forall i:$
	$\tar(e_i)=v_i$
or
$\forall i:$ 
	$\tar(e_i) = v_{i-1}$
and we call it a \emph{path from $s$ to $t$}
if it is a directed path between $s$ and $t$ with 
$\forall i:$
	$\tar(e_i)=v_i$.
The \emph{length} of a path is the number of its edges,
written as $|e_1,\dots,e_n| = n$.
The graph $G$ is said to be \emph{connected}
if there is a path between each pair of vertices.
In a path specification,
we may denote the vertices and edges
$a \in V \cup E$
by their respective label $\lab(a)$.
\end{definition}%
%
\begin{definition}[\tbd{Chord Path}]\label{def:chord}%
A \emph{chord path} of a cycle
shares its end vertices
with its corresponding cycle,
but none of its edges \cite{chordpath17, chordpath19}.
Given a graph $\dagdef$,
let the path $c = e_1,\dots,e_n$
with $e_i\!\in\!E$ for $i\!\in\![n]$,
be a cycle yielding $v_ne_1 \cdots e_nv_n$.
A \emph{chord~path} of the cycle $c$
is a path $e'_1,\dots,e'_m$ with
$e'_j\!\in\!E$ for $j\!\in\![m]$ yielding
$v'_0e'_1v'_1 \cdots e'_nv'_m$
with vertices $v'_0,\dots,v'_m\in\!V\!$ such that
$v'_0 \neq v'_m$ and
$\{ v'_0, v'_m \} =
  \{v_0, \dots , v_{n}\} 
   \cap
  \{v_0', \dots , v'_{m}\}$ 
but 
$\{e_i \:|\: i\!\in\![n] \}$ and
$\{e'_j\:|\: j\!\in\![m] \}$ being disjoint sets. 
\end{definition}%
\begin{definition}[\tbd{DAG, complete DAG, prefix-DAG}]\label{def:dag}%
A \emph{directed acyclic graph (over $\Gamma$)}, abbreviated as \emph{DAG},
is a graph over $\Gamma$ that does not contain any directed cycle.
The set of all connected and nonempty DAGs over $\Sigma$ is denoted by $\D_\Sigma$. 
A connected DAG $\dagdef$ is called a \emph{string DAG}
iff $|\IN(v)| \leq 1$ and $|\OUT(v)| \leq 1$ for all vertices $v \in V$.
Throughout this paper,
$\Sigma$ and $N$ being disjoint sets,
$\Sigma$ will denote an \emph{alphabet of terminals},
namely, the \emph{vertex labels},
whereas
$N$ is our \emph{alphabet of nonterminals}
used for labeling vertices and edges temporarily.
We call a DAG over $\Sigma$ a \emph{complete DAG}.
A DAG over $\Sigma \cup N$ is called a \emph{prefix-DAG},
a \emph{proper prefix-DAG} if at least one vertex
is labeled by a nonterminal $n \in N$.
\end{definition}


\begin{definition}[\tbd{Regular DAG grammar, $\Lgram$, $\Lall$ \cite{daggrammar}}]\label{def:grammar}%
A \emph{regular DAG grammar}\footnote{Since regular DAG grammars
are equivalent to regular DAG automata,
an illustrative example of how the DAGs are handeled 
can be found in \cite{DBLP:conf/mol/Drewes17}.} is a triple $\gramdef$.
Each \emph{rule} $r \in R$ is of the form $\alpha \sigmaarrow\beta$
where $\sigma \in \Sigma$ while
the \emph{head} $\alpha$ and the \emph{tail} $\beta$ are elements of $N^*$.
For the prefix DAGs $G$ and $G'$,
there exists a \emph{derivation step} $G\To_rG'$ using a rule $r$
if $G$ contains pairwise distinct vertices $v_1,\dots,v_k$
such that $\lab(v_1 \cdots v_k)=\alpha$.
In that case, $G'$ is obtained from $G$ by
\begin{itemize}
\item
  adding the vertex $v$ with its label $\lab(v)=\sigma$ 
\item
  by letting the edges, formerly pointing to  $v_1,\dots,v_k$,
  now point to $v$, thus $\tar(\IN(v_i)) = v$,
\item
  deleting the temporary vertices $v_1,\dots,v_k$ and, in turn,
\item
  adding the temporary vertices $w_1,\dots,w_j$ 
  labeled by their nonterminals $\lab(w_1\cdots w_j)=\beta$
\item
  by connecting them to the graph with the edges
  $(v, w_1),\dots(v,w_j)$.
\end{itemize}
A \emph{derivation} is a sequence of prefix DAGs%
\footnote{
We extend the notation $\To_{r_1\cdots r_n}$ to $\To_E$,
where $E$ is an (extended) regular expression over rules:
if $L(E)$ denotes the language of sequences of rules denoted by $E$,
then ${\To_E}={\bigcup_{r_1\cdots r_n\in L(E)}\To_{r_1\cdots r_n}}$.}
$G_0\To_{r_1}G_1\To_{r_2}\cdots\To_{r_n}G_n$, 
also denoted by $G_0\To_{r_1\cdots r_n}G_n$.
The set of all these $G_n$ that are complete is denoted by $\Lall$.
The \emph{DAG language generated by $\gram$} is 
$\Lgram = \{G\in\D_\Sigma\mid\emptygraph\To^*_R G\}$,
the set of connected and complete DAGs which the grammar can derive,
where $\To^*_R$ denotes the transitive 
and reflexive closure of ${\To_R}=\bigcup_{r\in R}{\To_r}$.
As usual, a rule is said to be \emph{useless}
if none of the derivations for DAGs in $L(\gram)$ comprises this rule
and \emph{useful} if it does.
%
The DAG grammar $\gram$ is \emph{deterministic} if,
for every pair $\alpha\sigmaArrow{}$ in $N^*\!\times\Sigma$,
there exists at most one $\beta\in N^*$
such that $(\alpha \sigmaarrow\beta)\in R$. 
The DAG language generated by a deterministic DAG grammar, 
as well as a DAG automaton recognizing it,
is called \emph{top-down deterministic}.
By reversing the orientation of the edges,
we obtain its \emph{dual language}.
A language $L(\gram)$ of a deterministic DAG grammar
is called \emph{bottom-up deterministic} if its dual language is generated by $\gram$.

The class $\RDL$ of \emph{regular DAG languages}
consists of all DAG languages generated by a regular DAG grammar (equivalently recognized by a regular DAG automaton).
The class $\RDLdet$ of \emph{regular deterministic DAG languages}
consists of all DAG languages and dual languages 
generated by regular deterministic DAG grammars (equivalently recognized by a regular deterministic DAG automaton). 

\end{definition}


In a derivation of a DAG $G$, at the time a new edge $e$ is generated,
its newly created target vertex $v$ is labeled by a nonterminal, say $q$.
At that time $v$ ``dangles'' at the end of $e$ without further incoming or outgoing edges.
Later rule applications will take $v$,
merge it with other vertices and 
label the resulting vertex
with its final symbol taken from $\Sigma$ according to the rule used.
The edge $e$, however, remains untouched.
We may represent a derivation of $G$ up to reordering of derivation steps by the DAG $G$ itself
together with a labeling of edges by nonterminals.
Then, $e$ would be labeled with $q$. We call this the derivation DAG of $G$.

\begin{definition}[\tbd{Derivation DAG, $\DAG D$}]\label{def:derivationdag}
Let $G_0\To_{r_1\cdots r_n}G_n$ with $r_1, \dots r_n \in R$ be
a derivation of a DAG $G_n = \dagdef$ generated by a DAG grammar $\gramdef$.
Then, the corresponding \emph{derivation DAG} of $G$
is the tuple $D=(V,E,\lab,\IN,\OUT)$,
where $\lab\colon E \cup V\to \Sigma\cup N$ is extended to edges by:
for every edge $e \in E$, $\lab(e)$
is the unique nonterminal $q \in N$ such that,
for some $i \in [n]$, $e$ is an edge of $G_i$ with $\lab(\tar_{G_i}(e)) = q$.
Let $\DAG D$ denote the DAG $G_n$, 
obtained from $D$ by restricting $\lab$ to $V$\!, denoted by $\lab|_V$.
\end{definition}


A derivation DAG is not necessarily connected,
thus $\DAG D \in \Lall$ if $\emptygraph \To^*_R D$
but $\DAG D \in \Lgram$ only if $D$, or equivalently $\DAG D$, is connected.
It should be noted that the set of all derivation DAGs of (DAGs in) $\Lgram$ is easily characterized:
For every such derivation DAG $D = (V,E,\lab,\IN,\OUT)$ the DAG $G = (V,E,\lab|_V,\IN,\OUT)$ is an element of $\D_\Sigma)$
and for every vertex $v\in V$ there is a rule $\alpha \sigmaarrow\beta$
such that $\alpha=\lab(\IN(v))$ and $\beta=\lab(\OUT(v))$.
Thus, the derivation DAGs of $\gram$ coincide with the runs of the DAG automata
in~\cite{journals/iandc/BlumDrewes2019},
and $\Lgram$ is the set of all DAGs $\DAG D$ such that $D$ is a derivation DAG of a DAG generated by $\gram$.
Moreover, a regular DAG grammar $\gram$ without useless rules is deterministic
if and only if every DAG in $\Lgram$ has exactly one derivation DAG.

\begin{definition}[\tbd{Rule Path and Cycle}]\label{def:rule}
\emph{Marking} a symbol $q$
(at a position $i \in [n]$ of a string of length $n$)
by a \emph{mark} $\m ~$ means replacing it with $\m q$.
We mark rules with the \emph{entry mark} $\men ~$
           and the \emph{exit mark} $\mex ~$;
if it is not specified which of those two marks is used, we use~$\m ~$.
A \emph{marked rule} $\m r =(\m\alpha\sigmaarrow\m\beta)$ is obtained from
a rule $r =(\alpha\sigmaarrow\beta)$ by marking two nonterminals
at two distinct positions $i, j \in [|\alpha\beta|]$ in $\alpha\beta$,
one with the entry, one with the exit mark;
in a \emph{weakly} marked rule only at one position
with either entry or exit mark.
Such a \emph{marked nonterminal} is referenced
by its tuple $(\m q, \m r)$ where
$\rmgrule$ is the (weakly) marked rule in which 
$q$ is replaced with $\m q$.
A rule pair $(\m r_i, \m r_j$)
for the (weakly) marked rules $\m r_i$ and $\m r_j$
\emph{agrees on} the marked nonterminals
$(\mex q, \m r_i)$ and $(\men q, \m r_j)$
if $q$ is marked
once in a head and once in a tail in order to obtain
$(\mex q, \m r_i)$ and $(\men q, \m r_j)$.
Two weakly marked rules always agree -- regardless of their marked nonterminals.
A \emph{rule sequence} is a nonempty sequence 
of (weakly) marked rules $\m r_1, \dots , \m r_n$
for which all rule pairs $(r_i , r_{i+1\mod n})$
with $i \in [n]$ agree 
and every marked nonterminal in the rule sequence is exactly once agreed on. 
A rule sequence of marked rules $\m r_1, \dots , \m r_n$
is a \emph{rule cycle}, with 
$\m r_1$ and $\m r_n$ being weakly marked
it is a \emph{rule path}.
A rule path \emph{between} $s$ and $t$
is a rule path with marked nonterminals $\m q_1, \dots ,\m q_{n-1}$ which
yields a path
$\sigma_1,q_1, \dots ,q_{n-1},\sigma_n$ in a derivation DAG
between $s$ and $t$
for $s\in \{q_1, \sigma_1\}$ and $t\in \{q_n, \sigma_n\}$.

\end{definition}

Figures~\ref{fig:rulecyclegraph} and
\ref{pic:bowrulecycle} show examples of rule sequences.
Observe that these definitions permit two types of rule sequences.
A rule sequence yields a directed path in a graph
if and only if both mark types,
exit and entry, do not occur both in heads and tails.
We call this a \emph{directed} rule sequence.
If heads (and consequently tails) comprise both types of marks,
the resulting path in the graph will be undirected.
This is called an \emph{undirected} rule sequence.

\begin{observation}[Yielding Cycle]\label{obs:yieldcycle}
Obviously, a directed rule cycle cannot yield a directed cycle in a DAG,
since DAGs are acyclic.
Therefore, only undirected rule cycles can yield a cycle in a DAG.
Directed rule cycles yield paths only, no cycles, in a DAG,
undirected ones yield both paths and undirected cycles.
\end{observation}

\begin{theorem}[Infinite Language \cite{journals/iandc/BlumDrewes2019}]\label{theorem:linfiffrc}
The DAG language generated by a \dagformalism{} without useless rules
is infinite iff $R$ contains a rule cycle.
\end{theorem}

\begin{definition}[Swap]
Let $G=(V,E,\lab,\IN,\OUT)$ be a DAG.
Two edges $e_0,e_1\in E$ are \emph{independent} if there is no directed path
between $e_0$ and $e_1$.
In this case, the \emph{edge swap of $e_0$ and $e_1$} is defined and
yields the DAG $G[e_0\bowtie e_1]=(V,E,\lab,\mathit{swap}\circ\IN,\OUT)$ given by the bijection
$\mathit{swap}\colon E\to E$ defined as $\mathit{swap}(e_i)=e_{1-i}$
for $i\in\{0,1\}$ and $\mathit{swap}(e)=e$ for $e\notin\{e_0,e_1\}$.
For $k \in\mathbb N$, let $G_0, G_1, \dots, G_k$
be $k+1$ disjoint isomorphic copies of $G$,
and for $i \in \{ 0, 1, \dots, k \}$,
let $e_i$ and $e'_i$ be the copies of $e$ and $e'$
in $G_i$, respectively.
Then the graph $\swapping Ge{e'}k$ for $k\in\mathbb N$ is defined as
\[\swapping Ge{e'}0 = G_0 
\text{\quad and\quad}
   \swapping Ge{e'}{k} =
  (\swapping Ge{e'}{k-1} \union G_{k})[e_{k-1}' \bowtie e_{k}].
\]
\end{definition}


Swapping two edges means that the tips of the arrows, the edge targets, are exchanged with one another. This swapping operation is, of course only allowed, if the result of the swap is still a DAG. Swapping edges yielding a directed cycle is not defined. The operation is central for regular DAG languages since, after swapping two edges in a DAG that have the same label in one of its derivation DAGs, the swapped result is still part of the language.
\begin{lemma}[\tbd{Swap Preserves Generation \refstandard}]\label{theorem:swap}
Let $\gramdef$ be a regular DAG grammar and
$D = (V,E,\lab,\IN,\OUT)$ a derivation DAG with $\DAG{D} \in \Lall$.
Then, if $\lab(e_0)=\lab(e_1)$ for $e_0, e_1\in E$,
the edge swap of $e_0$ and $e_1$ in $D$, in case it is defined, 
yields a DAG generated by $\gram$\!,
thus $ \DAG{D[e_0\bowtie e_1]} \in \Lall$.
\end{lemma}

Deterministic DAG grammars as defined above are equivalent
to the top-down deterministic DAG automata
in~\cite{journals/iandc/BlumDrewes2019}
and every derivation DAG of a grammar corresponds
one-to-one to a run of the corresponding DAG automaton.
Therefore, all results for top-down deterministic DAG automata
carry over to deterministic DAG grammars.
Refer to~\refstandard{} for the notation concerning DAG automata.
In particular, this holds for the following theorem,
in which a regular DAG grammar is called minimal if it does not contain useless rules
and there is no regular DAG grammar with fewer nonterminals generating the same language.

\begin{theorem}[Minimal Grammar \refstandard]
For every deterministic DAG grammar $\gram$,
a minimal deterministic DAG grammar $\gram'$ with $L(\gram') = L(\gram)$ 
can be computed in polynomial time.
This DAG grammar is unique:
every minimal deterministic DAG grammar that accepts $L(\gram)$
is identical to $\gram'$ up to a bijective renaming of its nonterminals.
\end{theorem}

\section{Meta-State}

Classical finite state automata recognize string languages by a finite memory,
its set of states.
In every step the automaton memorizes exactly one state.
On the contrary, while generating a DAG,
every derivation step of a DAG grammar
has to recall several nonterminals and not just one.
If we summarize those nonterminals to one meta-state per derivation step,
can we then accept DAGs with a finite state automaton? If so,
which kind of DAG languages can the finite state automaton recognize?

\begin{definition}[\tbd{Meta-state}]\label{def:metastate}
A multiset over nonterminals, i.e. an element of $\mathbb N^N$\!,
is called a \emph{meta-state}. 
The symbol $\metaQ$ is used for sets of meta-states,
thus $\metaQ \in 2^{\mathbb N^N}$\!.
For a DAG $\dagdef$ over $\Sigma\cup N$
we let $\meta G$ denote the meta-state $\lab_M(\{v\in V\mid \lab(v)\in N\})$,
the multiset of all labels which are nonterminals.\footnote{Note that,
according to the \emph{Notation} section,
the function $\lab_M$ returns a \emph{multiset} of labels.}
\end{definition}

Observe that 
it only depends on the meta-state of a graph, not on the graph as a whole,
how the derivation can proceed. For this we define the notion of a graph being useful.

\begin{definition}[Useful]\label{def:useful}
We say that a prefix DAG $G$
is \emph{useful} with respect to a given grammar $\gram$,
if this grammar $\gram$ derives $\emptygraph\To^*_RG\To^*_RG'$
with the DAG $G'$ being complete
and \emph{language-useful} 
if $G'$ is complete and connected, thus equivalently,
if $G$ occurs in a derivation of a DAG $G' \in \Lgram$.
\end{definition}

Note that a prefix DAG with respect to a grammar is useful 
if and only if all its connected components are language-useful.
Which properties depend on the meta-state only?

\begin{theoremEnd}[end, restate,text link=]{lemma}[\tbd{Meta-state dependent}]\label{theorem:metastate}
Let $G, G'$ be prefix DAGs with $\meta{G} = \meta{G}'$
derived by a DAG grammar $\gramdef$.
Then, $\gram$ can apply the rule $r \in R$ as the next derivation step
$G\To_rH$ iff $G'\To_rH'$.
Similarly, $G$ is useful if and only if $G'$ is useful.
However, if $G$ is language-useful this only implies that $G'$ is useful.
\end{theoremEnd}
\begin{proofEnd}
By the definition of a DAG grammar,
a rule application depends only on the temporary vertices labeled by non-terminals --
elements of the meta-states $\meta{G}$ and $\meta{G}'$.
Consequently, $G\To_rH$ is possible iff $G'\To_rH'$ is.
Valuability requires a rule sequence
and is therefore also dependent on the meta-state only.
It follows that $G$ is useful if and only if $G'$ is useful.
On the other hand, a common meta-state does not preserve connectivity.
With the given equality $\meta{G} = \meta{G}'$,
$G$ could be a connected DAG while $G'$ is not.
In that case $G$ is language-useful, but $G'$ not necessarily 
-- only if the preceding derivation connects the unconnected components of $G'$.
But since the language-useful $G$ is useful, $G'$ is, too.
\end{proofEnd}

Interestingly, a minimal grammar can finalize every derivation to a complete DAG.

\begin{theoremEnd}[text link=, end, restate]{lemma}[\tbd{Useful Prefix DAG}]\label{theorem:useful}
Let $G$ be a prefix DAG and $\gram = (N, \Sigma, R)$ be a DAG grammar without useless rules.
If $\emptygraph \To^*_R G$, then $G$ is \emph{useful with respect to} $\gram$,
i.e.\ there exists a derivation $\emptygraph \To^*_R G \To^*_R G'$ for a complete DAG $G'$.
If $G'$ is connected it is \emph{language-useful}.
\end{theoremEnd}

\begin{proofEnd}
We prove this by induction on the length $n$ of the derivation $\emptygraph \To_{r_1 \dots r_n} G$.
For the base case consider a derivation $\emptygraph \To_r G$ of length $n=1$.
As the rule set $R$ contains no useless rules,
there is a derivation $\emptygraph \To_{r'_1 \dots r'_k} G'$
where $G' \in D_\Sigma$ and $r = r'_i$ for some $i$.
Moreover, as $\emptygraph \To_r G$,
the head of the rule $r$ is the empty string and we may assume that $r = r'_1$.
This means that $G$ is useful with respect to $\gram$
(for the sake of brevity, we say useful subsequently).

Consider now a derivation $\emptygraph \To_{r_1 \dots r_n} G$ of length $n > 1$.
Let $\rho_1 = r_1 \dots r_{n-1}$.
By the induction hypothesis we know that there is a sequence $\rho_2$ of rules such that
$\emptygraph \To_{\rho_1 \rho_2} H$ for some $H \in D_\Sigma$.
Moreover, as $R$ contains no useless rules, there are rule sequences $\rho'_1, \rho'_2$,
such that $\emptygraph \To_{\rho'_1 r_n \rho'_2} H'$ for some $H' \in D_\Sigma$.
It is also possible to concatenate these two derivations and to interleave the individual derivation steps.
This yields a derivation $\emptygraph \To_{\rho_1 \rho'_1 r_n} H^* \To_{\rho_2 \rho'_2} (H \& H')$
where $H \& H'$ denotes the disjoint union of $H$ and $H'$.

Consider now the derivation $\emptygraph \To_{\rho_1 r_n} G$.
We want to show that $G$ is useful. 
As $\rho'_1$ can be applied on the empty graph, there is some DAG $G^*$
such that $\emptygraph \To_{\rho_1 r_n} G \To_{\rho'_1} G^*$.
Moreover, $G^*$ and $H^*$ have the same meta-state,
as they were generated through the same multiset of rules.
As $H^*$ is useful
it follows from Lemma~\ref{theorem:metastate},
that $G^*$ is useful, and therefore, $G$ is useful.

Therefore, if $G$ is connected,
so is the connected component $G''$ in the complete DAG $G'$
which finalized the prefix DAG $G$.
This means that $\gram$ can decide to choose only those derivation steps which yield a connected DAG $G' = G''$.
Consequently, a connected prefix DAG $G$
is not only useful
but also language-useful
with respect to its grammar $\gram$.
\end{proofEnd}

After having defined the notion of a meta-state, let us use them for derivations.
The first naive idea is to use all meta-states which occur in all the derivations for complete graphs.
We call that set $\metaQzero$ because a derivation of a DAG in the language
starts and ends with zero states. 
Apart from $\metaQzero$ also another set of meta-states is of interest.
While $\metaqzero$ incorporates all meta-states occurring during 
all derivations of DAGs in $L(\gram)$,
not all of these meta-states may be needed to generate the language $L(\gram)$.
This observation gives rise to a smaller set of meta-states $\metaQmin$. 

\begin{definition}[\tbd{$\metaQzero$ and $\metaQmin$}]\label{def:qmin}
Let $\gram$ be DAG grammar.
The set of all meta-states that occur in derivations of DAGs in $L(\gram)$ 
is denoted by $\metaqzero$ with
\[
\metaqzero=\{\meta G\mid\text{$G$ is a DAG which is language-useful with respect to $\gram$.}\}.
\]

A minimal set of meta-states is denoted by $\metaQmin$.
And, $\metaqmin$ denotes any set of meta-states such that
\begin{enumerate}
\item every DAG $G_n\in L(\gram)$ has a derivation
$\emptygraph\To_{r_1}G_1\To_{r_2}\dots\To_{r_n}G_n$
\\
such that $\meta{G_1},\dots,\meta{G_n}\in\metaqmin$, and
\item there is no set of meta-states of smaller cardinality with this property.

\end{enumerate}
\end{definition}

Thus, $\metaqmin$ is a minimal set of meta-states sufficient to generate a DAG language, while 
$\metaqzero$ incorporates also meta-states that could be dispensed.
The set $\metaqmin$ is not necessarily unique since often several derivations exist for one DAG,
and furthermore a permutation of derivation steps may result in different meta-states. 
In general, we are interested in $|\metaqmin|$, and in particular its finiteness, rather than in the set itself.

The subsequent example illustrates the existence of DAG grammars
for which $\metaqmin$ is finite while $\metaqzero$ is not.
\begin{example}[\tbd{DAG language of stars}]\label{ex:star}
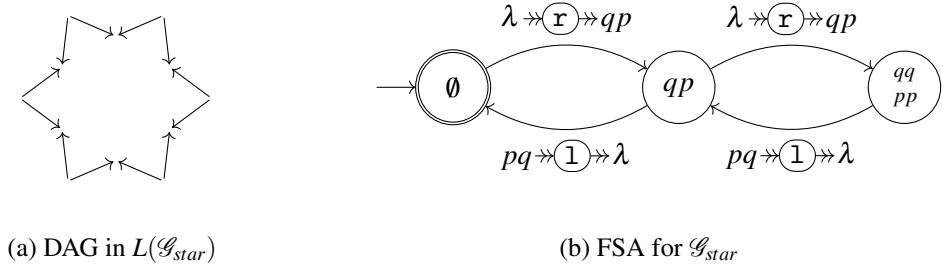
\begin{figure}[t]
\hfill
\subfloat[DAG in $L(\Gstar)$]{
  \label{pic:star}
  \tikzset{
shape example/.style= 
{
    line width = 0mm,
    inner xsep = 2.5cm,
    inner ysep = 0.5cm}
}
\begin{tikzpicture}
\pgfsetshortenstart{1pt}
\pgfsetshortenend{1pt}
\node[name=s, shape=star, star points=6, star point ratio=.65, shape example, inner sep=9mm] {};
\foreach \x in {1,...,6}
{
    \node [vertex] at (s.inner point \x) {};
    \node [vertex] at (s.outer point \intcalcInc{\intcalcMod{\x}{6}}) {};
    \draw (s.inner point \x) edge [->] (s.outer point \x);
    \draw (s.inner point \x) edge [->]  (s.outer point \intcalcInc{\intcalcMod{\x}{6}});
}
\end{tikzpicture}
  ~\\
}
\hfill
\subfloat[FSA for $\Gstar$]{
  \label{pic:astar}
  \begin{tikzpicture}[->,node distance=3cm,
    ever state/.style={inner sep=4cm},
    initial text=]
\node[initial,state,accepting] (0) {$\emptyset$};
\node[state] (r) [right of=0] {$qp$};
\node[font=\scriptsize,state, align=center] (rr) [right of=r,text width=3mm] {$qq$\\$pp$};

\path
( 0) edge [bend left]  node   [above]      { \starrone}         ( r)
( r) edge [bend left]  node   [above]      { \starrone}         (rr)
(rr) edge [bend left]  node   [below]      { \starrtwo}         ( r)
( r) edge [bend left]  node   [below]      { \starrtwo}         ( 0)
%
%
%
;
\end{tikzpicture}
  ~\\
}
\hfill
~
~\\
\captionsetup{subrefformat=parens}
\caption{The grammar $\Gstar$ gives rise to an FSA \subref{pic:astar}
that accepts DAGs like \subref{pic:star}
(\emph{labels are omitted}).}
\label{fig:star}
\end{figure}

Consider a DAG grammar $\Gstar$ with the rules
$\rr = \lambda \rarrow{r} qp$ and $\rl = pq \rarrow{l} \lambda $.
Let $G \in L(\Gstar)$,
cf.~Fig.~\ref{pic:star}.
First, $\emptygraph\To_{\rr^n} \Groot \To_{\rl^n} G$ for $n\ge 1$
generates a graph $\Groot$ consisting of $n$ roots labeled $\mathtt r$.
Subsequently, $\rl$
fuses pairs of nonterminal vertices into a single leaf labeled $\mathtt l$.
Collecting the meta-states that occur in these derivations
or in an arbitrary derivation $\emptygraph\ToR G$ both result in
$\metaQzero(\Gstar) = \{p^nq^n \mid n \in \mathbb{N}\}$,
since every rule either consumes or produces both a pair of nonterminals $q$ and $p$.
However, by first generating a single root, then alternating between $\rr$ and $\rl$,
and finally applying $\rl$ once more, $\Gstar$ offers the derivations $\emptygraph \To_{\rr(\rr\rl)^*\rl}\!\!G$
whose largest meta-state is $p^2q^2$.
Hence, $\metaQmin(\Gstar)=\{ pq, p^2q^2\}$.
\end{example}

The previous example gives rise to the following observation.
\begin{observation}\label{both}%
DAG Grammars $\gram$ with finite $\metaqmin$ but infinite $\metaqzero$ exist.
\end{observation}

This brings us to a further investigation of the finiteness of $\metaQzero$ and $\metaQmin$.

\section{Finite Number of Meta-States}

The previous section showed that languages generated with finite $\metaQmin$ indeed exist.
This section investigates which types of DAG languages can be generated
with a finite number of meta-states induced by the rules of a minimal deterministic grammar.
First, the newly identified language class deserves a name.
We call it \emph{finite induced meta-state DAG language}.
The term \emph{induced} is chosen since the grammar $\gram$ induces this set $\metaQ$ by
a suitable (c.f.~Lemma~\ref{theorem:both})
or all     (c.f.~Lemma~\ref{theorem:string-finite}) derivation DAGs.

\begin{definition}[\tbd{Finite induced meta-state DAG language ($\FID$)}]
\label{def:fidl}
A language recognized by a minimal deterministic grammar $\gram$
with finite $\metaqmin$ is called
a \emph{finite induced meta-state DAG language}.
The language class comprising finite induced meta-state DAG languages is denoted by $\FID$.
\end{definition}

In the following, we look into different categories for finite sets of meta-states.
For each category we check three types, first languages, second paths and finally rule cycles.
\begin{lemma}[\tbd{Finite $\Lgram$ -- finite $\metaQzero$ and $\metaQmin$}]\label{theorem:finite}
\begin{enumerate}
Let $\gramdef$ be a DAG grammar.\\
\item 
If $L(\gram)$ is finite,
the sets $\metaqzero$ and $\metaqmin$ are finite as well.
\item\label{fin:item:path}
For a rule path $\Pi$ of finite length within $R$,
whose marked rules comprise only marked nonterminals
except for the weak marked rules,
finitely many meta-states suffice for all derivation sequences
which generate the corresponding path $\pi$ of $\Pi$ in a DAG.
\end{enumerate}
\end{lemma}
\begin{proof}
Since $N$ is finite and $L(\gram)$ is finite, only a finite number of derivations for all DAGs in $L$ exists.
Obviously, this combination yields a finite number of meta-states for $\metaqzero$ and consequently also for $\metaqmin$.
In the second statement, the finite rule path $\Pi$ 
gives rise to a finite derivation $G \ToPi G'$ for all possible prefix DAGs $G$ of $G'$,
since, due to all nonterminals marked except start and end,
$\Pi$ does not need to be interleaved with rules.
A finite derivation yields finitely many meta-states, both for $\metaQmin$ as well as for $\metaQzero$.
\end{proof}

Infinite languages $\Lgram$ do not necessarily induce
an infinite $\metaQmin$ and not even an infinite $\metaQzero$.
First, we look at those cases, where they indeed induce finite sets of meta-states only.
\begin{lemma}[\tbd{Strings -- finite $\metaQzero$ and $\metaQmin$}]\label{theorem:string-finite}
\begin{enumerate}
Let $\gramdef$ be a DAG grammar. 
\item\label{string:item:lang}
If $L(\gram)$ is a string DAG language
the sets $\metaQzero(\gram)$ and  $\metaqmin$ are finite
\item\label{string:item:path}
For a, possibly arbitrary long, directed rule path, $\Pi$ in $R$,
whose marked rules comprise only marked nonterminals except for the weak marked rules,
finitely many meta-states suffice for 
all derivation sequences which generate the corresponding path $\pi$ by $\Pi$ in a DAG.
\item\label{string:item:cycle}
For a directed rule cycle $c$
whose marked rules comprise only marked nonterminals,
finitely many meta-states suffice for all derivation sequences
which generate the corresponding path $\pi$ by $c$ in a DAG.
\end{enumerate}
\end{lemma}
\begin{proof}
Since string languages consume and produce exactly one nonterminal in every derivation step
except for the first and last steps producing and consuming one nonterminal only, respectively,
$N$ equals $\metaqzero$.%
\footnote{Since we use strings as a notation for multisets, this special set of multisets equals a normal set.}
Obviously, this means that $\metaqzero$ is finite, since $N$ is finite.

This carries over to subgraphs which are string DAGs.
Strings as subgraphs are derived in the statements (\ref{string:item:path}) and (\ref{string:item:cycle}).
Let $\m r$ denote a marked rule with only marked nonterminals of the form $\men q\sigmaarrow \mex p$.
Such a rule $\m r$ does not alter the size of the meta-states,
thus $|\meta G| = |\meta G'|$,
in its derivation step $G \Tor G'$.
For a directed rule cycle $c$ with only marked nonterminals holds the same
just as for an infinite rule path $\Pi$ in $R$,
since they consist of a rule sequence with rules of type $\m r$.
Such a rule path $\Pi$ comprises only finitely many marked rules since $R$ is finite.
Thus, an arbitrary long rule path corresponds to a rule cycle.
Consequently, $\metaQzero$ is finite for both (\ref{string:item:path}) and (\ref{string:item:cycle}).


In all three cases $\metaQzero$ is finite, yielding a finite $\metaQmin$.
\end{proof}

In summary, for strings of any kind, string languages or strings as subgraphs, also of unbounded length,
both $\metaQmin$ and $\metaQzero$ stay finite. This was to be expected since regular string languages
are accepted by finite state automata with a finite set of states.
What happens beyond string DAGs?
\begin{lemma}[\tbd{Finite $\metaQmin$ -- infinite $\metaQzero$}]\label{theorem:both}

\begin{enumerate}
Let $\gramdef$ be a minimal deterministic DAG grammar 
\item In the case where $\metaqmin$ is finite, it is possible that $\metaqzero$ is infinite. 
\item\label{both:item:path}
For a possibly arbitrary long, undirected rule path, $\Pi$ in $R$,
whose marked rules comprise only marked nonterminals except for the weak marked rules,
finitely many meta-states suffice for deriving all complete DAGs
incorporating the corresponding path $\pi$ generated by $\Pi$,
however, some derivations for each such complete DAG ask for an infinite set of meta-states.
\item\label{both:item:cycle}
For an undirected rule cycle $c$ in $R$,
whose marked rules comprise only marked nonterminals,
finitely many meta-states suffice for deriving all complete DAGs
incorporating the corresponding path $\pi$ of unbounded length generated by $c$.
However, some derivations for each such complete DAG ask for an infinite set of meta-states.
\end{enumerate}
\end{lemma}
\begin{proof}
The first statement is equivalent to Observation~\ref{both}.
Both Statements (\ref{both:item:path}) and (\ref{both:item:cycle}) may yield a path $\pi$ of unbounded length.
Such a path $\pi$ gives rise to a derivation
which does not increase the cardinality of the meta-state,
such that $\metaQmin$ with respect to $\Pi$ and $c$ is finite.
This is possible by using rules
that add nonterminals to the current meta-state
only when they are needed,
just like for the DAG language of stars in Example~\ref{ex:star}.
But for undirected rule paths and cycles, we can rearrange
the rules in the derivations
in order to first generate all roots (again, c.f. Example~\ref{ex:star}).
Since both $\Pi$ and $c$ can make the path $\pi$ unbounded,
the size of the meta-states will grow without bound,
yielding an infinite $\metaQzero$.
\end{proof}

In the next section, we will turn to languages which cannot be generated with a finite set of meta-states.
Prior to this, we look at the unexpected fact that 
although a label does not occur in any rule cycle, its number of occurrences could be unbounded.

\newcommand\captiondecstarmany{$G_k[e_{k-1} \bowtie e_k]=G_{k-1}\union G_{k}$}
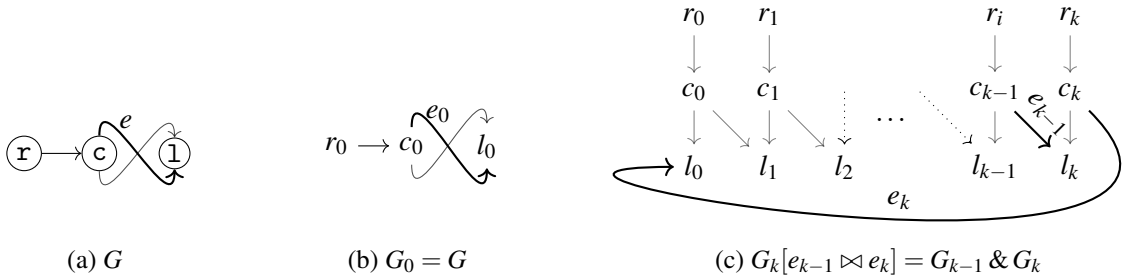
\begin{figure}
\vspace{-.8cm}
\centering
%
  \subfloat[$G$]{
    \label{pic:decstarG}
    \begin{tikzpicture}[->,auto=right]
\tikzstyle{vx}=[circle,draw,inner sep=2pt]
\node [vx]            (r) at (0,0) {$\labelr$};
\node [vx,right of=r] (m)          {$\labelc$};
\node [draw,circle,inner sep=1pt,right of=m] (l)          {$\labell$};

\path (r) edge (m)
      (m) 
          edge [out= 90, in=270,looseness=1.8,thick] node [pos=.4,above] {$e$}(l)
      (m) edge [out=270, in= 90,looseness=1.8,opacity=.7] (l);

\end{tikzpicture}
  }
  \hfill
  ~
  \hfill
  \subfloat[$G_0 = G$]{
    \label{pic:decstarG0}
    \begin{tikzpicture}[->,auto=right]
\node              (r) at (0,0) {$r_0$};
\node [right of=r] (m)          {$c_0$};
\node [right of=m] (l)          {$l_0$};

\path (r) edge (m)
      (m) edge [out= 90, in=270,looseness=1.8,thick] node [pos=.4,above] {$e_0$}(l)
      (m) edge [out=270, in= 90,looseness=1.8,opacity=.7] (l);

\end{tikzpicture}
  }
  \hfill
  ~
  \hfill
  \subfloat[\captiondecstarmany]{
    \label{pic:decstarmany}
    \begin{tikzpicture}[
->,
decoration/.style={pin={above:m}},
unemph/.style={opacity=.5},
pin distance=4pt,
every pin edge/.style={<-,},
]
  \node                       (c0)    {$c_0$};
  \node       [right of=c0]   (c1)    {$c_1$};
  \coordinate [right of=c1]   (c2);
  \coordinate [right of=c2]   (ci-1);
  \node       [right of=ci-1] (ci)    {$c_{k-1}$};
  \node       [right of=ci]   (ck)  {$c_k$};
  \node       [below of=c0]   (l0)    {$l_0$};
  \node       [below of=c1]   (l1)    {$l_1$};
  \node       [below of=c2]   (l2)    {$l_2$};
  \node       [below of=ci]   (li)    {$l_{k-1}$};
  \node       [below of=ck] (lk)    {$l_{k}$};

  \foreach \c in {0, 1, i, k}
  {
      \node [
             above of=c\c] (r\c) {$r_\c$};
      \path[->, unemph] (r\c) edge 
                        (c\c);
  }
  \path[->] (c0) edge [unemph] (l0)
            (c0) edge [unemph] (l1)
            (c1) edge [unemph] (l1)
            (c1) edge [unemph] (l2)
            (ci) edge [unemph] (li)
            (ci) edge [thick] node [sloped,above] {$e_{k-1}$} (lk)
            (ck) edge [unemph] (lk);
  \path[->,dotted] (c2)   edge (l2)        
                   (ci-1) edge (li);
  \path (ck) edge[->,thick,out=-45,looseness=1.8,in=180] node [above, pos=0.45] {$e_k$} (l0);
  \path (l2) -- node {$\dots$} (ci-1);

  \pgfresetboundingbox
  \node [below  left of=l0] (helpnode1) {};
  \node [above right of=rk] (helpnode2) {};
  \path [use as bounding box] (helpnode1) rectangle (helpnode2);

\end{tikzpicture}
  }
  \hfill
  ~
\captionsetup{subrefformat=parens}
\caption{The decorated one-pointed star $G$ is shown in \subref{pic:decstarG}.
This equals the star $G_0$ in \subref{pic:decstarG0},
where in addition to the vertex label,
the label's index allows us to reference  
each vertex uniquely by its number of copies. 
Note thus, that in \subref{pic:decstarG0},
as well as in~\subref{pic:decstarmany},
the index is not part of the vertex label.
To draw the graph itself in those two pictures,
the indices would be stripped off.
Note, \subref{pic:decstarmany} illustrates the swapping
of $k+1$ disjoint isomorphic copies
to a $k+1$-pointed star $G(e \bowtie e)^k$.
\label{fig:noncycle}}
\end{figure}

\newcommand\ellrule{l \rarrowl\lambda}
\begin{theoremEnd}[text link=,end, restate]{lemma}[\tbd{Unbounded without Rule Cycle}]\label{theorem:noncycle}
There exists a DAG grammar $\gramdef$ and
a vertex or edge label $\labelu \in \Sigma \cup N$
such that the number of occurrences of $\labelu$
in a (derivation) DAG generated by $\gram$ 
is unbounded,
although $\labelu$ does not occur in any rule cycle.
\end{theoremEnd}
\begin{proofEnd}
Consider the DAG grammar
$\gram = (\{r, c ,l\}, \{\labelr, \labelc, \labell\}, R )$ containing the following rules
$R = $
\mbox{
$\{
  \lambda \rarrow{r}rr,\allowbreak\ 
  rr \rarrow{c} c,
  c \rarrow{c} l,\allowbreak\ 
  \ellrule
\}$
}
and let $\labelu \in  \{ l, \labell\}$,
thus let $\labelu$ either denote
the edge label $l$ or the vertex label $\labell$.
Then, the vertex label $\labell$
cannot occur in a \emph{marked rule} $\m\alpha\sigmaarrow\m\beta$,
which, by definition, comprises two marked nonterminals.
On the contrary, the only rule with the label $\labell$ is 
$\ellrule$ with $|\alpha\beta|\!=\!|l\lambda|\!=\!1$
and thus comprises only one nonterminal.
Consequently, $\labell$ cannot occur in a rule cycle since 
a rule cycle contains marked rules only.
And, neither can the edge label $l \in N\!$.
In a rule cycle, $l$ would occur both in a head as well as in a tail
which it does in $R$, but, again,
the rule $\ellrule$ cannot take part in a rule cycle.
Thus, by definition, 
$\labelu$ cannot participate in any of $\gram$'s rule cycles.

The DAG $G \in \Lgram$ in Figure~\ref{pic:decstarG}
uses each rule $r \in R$ only once,
just as both vertex label~$\labell$ and edge label $l$.
We can take $k$ disjoint isomorphic copies of $G$
for any $k \in \mathbb{N}$
and connect them by swapping the copies of $e$,
as shown in Figure~\ref{pic:decstarmany},
by Theorem~\ref{theorem:swap}.
The resulting DAG $G[e_{k-1} \bowtie e_k]^k$
is still accepted by $\gram$ and connected.
Moreover, it contains $k$~occurrences of label $\labelu$,
which proves the lemma.
\end{proofEnd}

The following lemma summarizes when the number of label occurrences is not bounded.

\begin{theoremEnd}[end, restate,text link=]{lemma}[\tbd{Unbounded Label Occurrence}]\label{theorem:unbound}
Let $\gramdef$ be a minimal deterministic DAG grammar
and $\labelu\in\Sigma\cup N$ a label of a vertex or an edge.
The number of occurrences of $\labelu$ in graphs $G \in \Lgram$,
or, for edge labels, in their corresponding derivation DAGs $D$, is unbounded,
iff one of the two following conditions is fulfilled:
\begin{enumerate}
\renewcommand\theenumi{\alph{enumi}\normalfont{)}}
\renewcommand\labelenumi{\theenumi}
\item\label{unbound:cycle}
Label $\labelu$ occurs in some rule cycle of $\gram$.
\item\label{unbound:noncycle}
There exist both a rule cycle $c$ in which an unmarked $q\in N$ occurs
as well as a rule path~$\Pi$ between this nonterminal $q$ and the label $\labelu$.
\end{enumerate} 
\end{theoremEnd}

\begin{proofEnd} 
Suppose that \ref{unbound:cycle} is true.
Then $L(\gram)$ is infinite, according to Theorem~\ref{theorem:linfiffrc}, which in turn means
that there is no bound on the length of (possibly undirected) paths in graphs $G \in L(\gram)$.
To obtain paths of unlimited length by the pigeonhole principle the repetition of rules is needed since $R$ is finite.
We may do so by using the rule cycle $c_u$ comprising $\labelu$. 
The derivation $\emptyset \To_R G' \To_{c_u}^* \To_R G$ does not impose a bound
on the number of label $\labelu$ occurrences on $G$s derived like that,
which proves that \ref{unbound:cycle} implies unboundedness of $\labelu$ occurrences.

Suppose that  \ref{unbound:noncycle} is true.
The cycle $c$ with the unmarked state $q$ generates not only the labels it comprises arbitrarily often,
as showed in above paragraph, but also an unbounded number of the nonterminal $q$ may appear
when taking the intermediate graphs into account that are generated in the various steps to yield $G$.
At every $q$ the derivation steps $\To_{\Pi}$ generate
a useful DAG with respect to $\Lgram$ with a path comprising $\labelu$.
But $\gram$ does not bound the generations of $q$, thus neither on the derivation step $\To_{\Pi}$
and consequently also the number of the labels $\labelu$s which shows that 
\ref{unbound:noncycle} implies that the number of occurrences of $\labelu$ is not bounded.

Turning now to the second direction we assume that no bound on the number of $\labelu$s
exists for the DAG language $L(\gram)$.
This means that $\gram$ can repeatedly generate $\labelu$. 
With a finite number of rules $R$, this is only possible by applying rules with label $\labelu$ an unbounded number of times.
Repetition of rules is obtained either by a rule cycle,
\ref{unbound:cycle} or by rule paths to a cycle \ref{unbound:noncycle}.
In every iteration of the cycle, also the rule path $\Pi$ is repeated and like that our label $\labelu$.

Obviously, rules that do not participate in any rule cycle,
not having a path to a rule cycle cannot be used in a derivation multiple times,
which proves the second direction.
This result carries over without difficulty when regarding $u$
as the edge label $u \in N$ in a derivation DAG $D$,
instead of the vertex label $u \in \Sigma$ of a DAG $\DAG D = G$,
which completes the proof.
\end{proofEnd}

\section{Infinite Number of Meta-States}

Unfortunately, from the algorithmic viewpoint at least,
there exist infinite DAG languages where a finite number of meta-states
is not sufficient for generating them in a top-down deterministic manner.
We call this language class the infinite meta-state DAG languages, abbreviated as $\ID$.

\begin{definition}[\tbd{Infinite meta-state DAG language ($\ID$)}]\label{def:idl}
A language $\Lgram$ of a minimal deterministic DAG grammar $\gram$ is called an
\emph{infinite meta-state DAG language} if $\metaqmin$ is infinite.
The language class comprised of all infinite meta-state DAG languages is denoted $\ID$.
\end{definition}

\begin{lemma}[\tbd{Infinite $\metaQmin$ -- infinite $\metaQzero$}]\label{theorem:infinite}
Infinite meta-state DAG languages exist.
\end{lemma}
\begin{proof}
Let us consider the following DAG grammar for binary trees
$\Gtree = (\{q\},$ $\{\rsymbol,\msymbol,\lsymbol\}, R)$.
As a tree, every $G \in L(\Gtree)$ has only one root. Therefore, every derivation for the tree grammar $\Gtree$
applies $r_r = \lambda \rarrow{r} qq$ only once: $\emptygraph \To_{r_r} \groot \To_{R \setminus \{r_r\}} G$.
The leaf rule $\lrule = q \rarrow{l} \lambda$ cannot be part of a rule cycle
since this requires a minimum of two nonterminals.
The rule $\mrule = q \rarrow{m} qq$ is the only rule that occurs in a rule cycle.

Every derivation step using $\mrule$ increases the cardinality of the meta-state by one
due to consuming one nonterminal and producing two nonterminals.
If $\mrule$ is the only rule in a rule cycle and $\mrule$ increments the meta-state
in every derivation step $ \To_{\mrule}$,
do we end up with an infinite set $\metaQmin$?
Luckily, Lemma~\ref{theorem:noncycle} tells us that we can repeatedly apply derivation
steps also with rules not being part of a rule cycle.
Lemma~\ref{theorem:unbound} confirms that this is the case for the rule $\lrule$,
since we have an unmarked $q$ in all marked rule combinations of $\mrule$
for a rule cycle and we have a rule path $q$ to $\labell$ only with the rule $\lrule$.
The number of leaves, labeled by $\labell$, is not bounded.
As a sole rule in $R$ the leaf rule decreases the cardinality of the meta-state.
It is folklore that the amount of leaf vertices is one more than non-leaf vertices.
This sounds thus promising to end up for a finite set $\metaQmin$.
The derivations 
$\emptygraph \To_{r_r} \groot \To_{\lrule} \groot' \To_{(\mrule \lrule)^*} G' \To_{\lrule} G$
generate DAGs in $L(\gram)$.
Those derivations include only two meta-states
$\metagroot = q^2$ and $\metagroot = \meta G'= q$.
Unfortunately, this does not yield fully balanced trees.

If we consider the fully balanced binary trees of depth $n$,
where every leaf is the $n$th vertex in a path from the root to the leaf,
$\gram$ has to generate at least $n-2$ vertices labeled $\labelm$ or $n-1$ vertices not labeled $\labell$
until it can apply the rule $\lrule$. This, in turn, means that $n-2$ derivation steps increment the meta-state.
The generation of a complete binary tree results in a meta-state $q^n$.
For generating the language of binary trees of all depths, $\metaQ = \{q^i \mid i\!\in\![n]\}$.
But this set is minimal, since no other derivations exist,
and infinite which shows that $\metaQmin(\Gtree)$ is infinite.
Note, $\Gtree$ is a deterministic and minimal grammar.
Thus, minimal deterministic DAG grammars with finite $\metaQmin$ exist.
\end{proof}

After proving the existence of languages in the language class $\ID$,
we present languages in this class.
The following theorem about certain tree languages follows directly
from Lemma~\ref{theorem:infinite}. 

\begin{theorem}[Trees in $\ID$]\label{theorem:tree2id}
A minimal deterministic DAG grammar $\gramdef$
generating fully balanced trees of unbounded size,
$\Lgram$ is in $\ID$.
\end{theorem}

\begin{proof}
For $\Lgram$ comprising fully balanced trees with one root,
$\Lgram \in \ID$ was shown to be true in Lemma~\ref{theorem:infinite}.
For trees ``upside down'', one leaf and many roots,
the argumentation is similar.
Instead every derivation step $\To_{\mrule}$
for a path from a root to a leaf
adds a state to produce the ensuing meta-state.
In doing so, it forces dependencies to be derived.
The top-down generation of $\gram$ prevents the isolated derivation of a root-to-leaf path.
A vertex at depth $n$ requires $n$ roots.
Therefore, the size of the meta-states depends on $n$, yielding infinitely many meta-states.
Thus, the dependencies are responsible for the infinite class $\metaqmin$.
\end{proof}

Trees are fine, but what about graphs?
We observed (Obs.~\ref{obs:yieldcycle}) that only undirected rule cycles can yield cycles in DAGs.
Directed trees are generated by directed rule cycles. So let's consider
languages with proper graphs, thus generated by undirected rule cycles.

\begin{lemma}[DAG Cycles in $\ID$]\label{lemma:chord2id}
Let $\gram$ be a minimal deterministic DAG grammar.
Then $\Lgram \in \ID$ if some DAG $G \in \Lgram$ contains
a cycle of arbitrary length that exhibits a chord path.
\end{lemma}

\begin{proof}
%
As Lemma~\ref{theorem:infinite} tells us,
generating a cycle of arbitrary length causes $\metaqzero$ to be infinite,
but not necessarily $\metaqmin$.
Hence, the infinite size of $\metaqmin$
hides in the rule cycle's chord path.

In a DAG, a cycle of arbitrary length $m \in \mathbb N$
requires the repeated application 
of an undirected rule cycle (Obs.~\ref{obs:yieldcycle})
in an unbounded number of iteration steps,
say, in $n \in \mathbb N$ iteration steps,
since the set of rules in $\gram$ is finite and,
in order to obtain a cycle of arbitrary length $m$, as stated, 
by the pigeonhole principle,
we need an unbounded number $n$ of applications of the required rule cycle.

\newcommand\csize{|c|}
Let $c$ be this undirected rule cycle that is required.
With the constant $\csize$ denoting the number of rules $c$ consists of,
the length of the DAG's undirected cycle is $m = \csize \cdot n$. 

Consider the DAG $G$ generated by $2n$ rule cycle iterations of $c$,
yielding a cycle of length $2\cdot\csize\cdot n$,
as detailed in the following.

%
As stated, the cycle in $G$ has a chord path.
However, if one application of $c$ generated the stated chord path,
a graph generated by many applications of $c$,
could, if `wired' appropriately,
contain many such chord paths instead of just one.
The grammar $\gram$ generated $G$ by $2n$ applications of $c$.
Consequently, $\gram$ can generate $G$ in a way
that it exhibits $2n$ such chord paths,
thus, with identical edge labels as the stated chord path exhibits.
The identical labeling is key for the proof,
since this will allow `wiring' the chords differently.

%
By definition (Def.~\ref{def:rule}),
a rule cycle includes marked rules
with a minimum of two nonterminals with two of them marked.
However, obviously, the end vertices of a chord path
have a vertex degree of at least three \cite{textbooks/graph/westchordpath}.
Due to that, the marked rules generating the chord path
need at least three nonterminals and that means one unmarked nonterminal.
Thus, a chord path requires an unmarked nonterminal.
An unmarked nonterminal needs to be \emph{saved} in a meta-state.
Let us call this nonterminal $s$. 
Assume that $\gram$  opens and closes the chord paths one after another
while generating $G$.
Then $\gram$ could possibly reuse the meta-state in which we stored $s$:
when it opens the chord path, it stores $s$ in the ensuing meta-state,
when it closes the chord path, it does not need to save $s$ anymore.
Then, $\gram$ repeats this step for the next chord path,
reusing the meta-state in which it stored $s$. 

Yet, Lemma~\ref{theorem:swap} allows the swapping of the $2n$ chord paths.
In order to obtain one of those graphs in $\Lgram$ 
requiring an infinite $\metaqmin$, let us swap the chord paths in $G$
(technically more precise: certain edges of the chord paths).
Let $e_i$, labeled with the afore mentioned $\ell(e_i)=s$,
be one of the ends of the chord path number $i$ for $i \in [2n]$.
Swapping $G[e_1\bowtie e_n]$, $G[e_2\bowtie e_{n-1}]$, \dots,
thus $G[e_{k}\bowtie e_{n-k+1}]$ for $k \in [n/2]$
results in the first $n$ chord paths bridging a distance of length $\csize\cdot n$ of its DAG cycle.

Please note: Swapping the second half of the chord paths is unnecessary,
but can be performed in order to obtain a symmetric, beautiful graph.
If we would, however, swap the $2n$'th chord path with the first, the distance would be one length of $c$, since the chords are arranged on a cycle.
First and $2n$'th chord paths are neighbors, thus.
That is the reason for choosing $2n$ instead of $n$ cycle iterations of $c$.
The first and the $n$'th chord paths are a maximum distance apart
with respect to the cycle generated by $2n$ applications of~$c$.

After these swapping operations, the chords are nested:
$\gram$ opens the first chord path
before the second but closes the second before the first,
and so on, according to the FILO (first in, last out) principle.
This constitutes the proof: by opening $n/2$ chord paths,
all identically labeled with $s$,
there exists a meta-state in $\metaqmin$
that has to store all those $s^{n/2}$ labels. 
Then, $\gram$ has to remember
that all these edges of the chord paths are dangling.
Thus, the number of meta-states depends on
the number of cycle iterations $2n$,
which completes the proof, since $n \in \mathbb N$.
\end{proof}


\section{Characterization}

Now we contemplate the newly found language classes which divide the class
of top-down deterministic regular DAG languages $\RDLdet$,
whose languages are generated by deterministic DAG grammars,
into disjoint subclasses.

\begin{observation}[$\FID$ $\cup$ $\ID$ $= \RDLdet$]\label{obs:fididrdldet}
By the definitions of $\FID$ and $\ID$, it is true that $\FID$ and $\ID$ are disjoint and that $\FID$ $\cup$ $\ID$ $= \RDLdet$.
\end{observation}

The following lemma characterizes the language class $\ID$ by the rules of a grammar.
Since $\ID$ and $\FID$ are disjoint, indirectly it is also a characterization of $\FID$.

\begin{figure}
\vspace{-.8cm}
  \begin{subfigure}[b]{0.3\textwidth}
    \centering
    \begin{tikzpicture}[inner sep=-1pt, outer sep=0pt,minimum size=0]
\node (m) at (0,0) [draw, circle,inner sep=2pt] {~};
\node (helpernode) at (0,0.5) {};
\node [below left of=m] (helpernodeleft) {};
\node [below right of=m] (helpernoderight) {};

\path (helpernode) edge [-latex] (m)
      (m.south west) edge [bend left] (helpernodeleft)
      (helpernodeleft) edge [bend left,out=109,in=55] node[left] {$c$\hspace{4pt}~} (helpernode.north)
      (m.south east) edge [bend right, dotted, thick] (helpernoderight)
      (helpernoderight.east) edge [bend right,out=-109,in=-55, dotted, thick] node[right] {\hspace{2pt}~$\Pi$} (helpernode.north)
      ;
\end{tikzpicture}
    \caption{Rule cycles' graph}
    \label{pic:rulecyclegraph}
  \end{subfigure}
  \hfill
  \begin{subfigure}[b]{0.3\textwidth}
    \centering
    $
\anchor{h}{\men q}
\earrow
\anchor{t}{\mex q}
q
$
\begin{tikzpicture}[remember picture,overlay,anchor=base]
        \path (h) edge[latex-,out=90,in=90,looseness=1.4] (t);
\end{tikzpicture}
    \caption{Rule cycle $c$ -- left branch}
    \label{pic:rulecycle-tree-l}
  \end{subfigure}
  \hfill
  \begin{subfigure}[b]{0.3\textwidth}
    \centering
            $
        \anchor{r1h0}{q}
        \earrow
        q
        \anchor{r1t1}{\mex q}
    \qquad
        \anchor{r2h0}{\men q}
        \earrow
        q
        \anchor{r2t1}{q}
        $
        \begin{tikzpicture}[remember picture,overlay]
                \path (r1t1) edge[-latex, dotted, thick, out=90,in=90,looseness=2] (r2h0);
        
        \end{tikzpicture}
    \caption{Rule path $\Pi$ -- right branch}
    \label{pic:rulecycle-tree-r}
  \end{subfigure}
\caption{Rule cycles and their graphs for the binary tree language given by a cycle and its chord}
\label{fig:rulecyclegraph}
\end{figure}
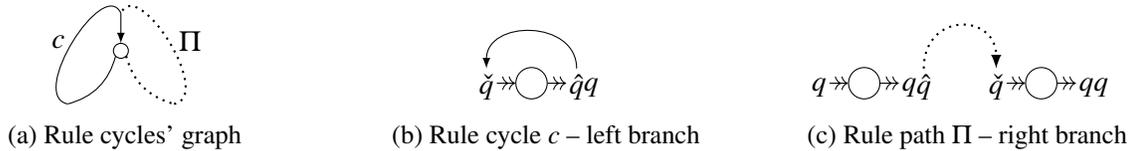

\begin{theorem}[\tbd{Characterization of $\ID$}]\label{theorem:charrule}
For a minimal deterministic DAG grammar $\gram = (N, \Sigma, R)$
its set of meta-states $\metaqmin$ is infinite iff there exist
a rule cycle $c$, a rule path $\Pi$
and not necessarily distinct nonterminals $q, p \in N$
satisfying the following conditions:

\begin{itemize}
\item\label{charrule:item-c}
The nonterminals $q$ and $p$ occur in $c$ as unmarked $q \in N$ and marked $\m p \in \m N$.\lipicsrightdescr{$c$}
\item\label{charrule:item-Pi}
The rule path $\Pi$ lies between $q$ and $p$.
\lipicsrightdescr{$\Pi$}
\item\label{charrule:item-r}
The nonterminal $q$ is in $c$ in a head iff it is in $\Pi$'s weak rule in the head.
\lipicsrightdescr{$c\Pi$}
\end{itemize} 

\end{theorem}

\begin{proof}
Before proving the stated bi-implication, observe 
that the above conditions result in two distinct rule cycles.
It is easy to see that $\Pi$ acts as a chord path for the rule cycle $c$,
and consequently, provides an alternative rule cycle.
Fig.\,\ref{fig:rulecyclegraph} demonstrates an example.
For the rule cycle $c$ condition
      \lipicsfont{$c$}
and
for the rule path $\Pi$ condition
      \lipicsfont{$\Pi$}
urge the suitable nonterminals for gluing $\Pi$ as a chord to $c$.
The requirements stated in
      \lipicsfont{$c\Pi$},
concerning the nonterminals' positions in head or tail
ensure the right orientation of the edge labeled $q$.
And, this alternative rule cycle provided by $\Pi$ is the crux of the matter.
Since the grammar has always two options to continue the derivation,
it always has to remember the dangling edges of the alternative it did not decide to generate next.
Thus, in every cycle iteration, regardless of the decision how to proceed,
a new meta-state is needed since the number of dangling edges
grows in an unbounded fashion.

We start with assuming that $\gram$'s rules have the stated form
and then prove that $\metaqmin$ is infinite.
The proof by construction relies on two techniques.
We construct a graph $G_n$
for which the number of meta-states in $\metaqmin$ needed to generate it
depends on the number of rule cycle iterations $n \in \mathbb{N}$.
The first technique is the (directed or undirected) path whose length depends on $n$:
Our grammar $\gram$ can generate the above stated DAG $G_n$
since a rule cycle can yield a path of length $k \cdot n$
where $k$ is the number of edges generated by the marked nonterminals in one cycle iteration.
A path as such is connected, and our path is a subgraph of a connected prefix DAG, generated top-down by $\gram$.
As such, it is a connected prefix DAG which, according to Lemma~\ref{theorem:useful},
$\gram$ can complete to a language-useful DAG.
The second technique is to disregard all the edges and vertices
not mentioned in the conditions stated about the rules of the grammar.
Obviously, they are irrelevant since, due to the swapping operation (Theorem~\ref{theorem:swap}),
they would only \emph{increase} the size and number of the meta-states.\footnote{%
Note that, with this second technique applied, $k$
-- the number of edges generated by the marked nonterminals in one cycle iteration --
equals to one. 
We abstract away how many edges the cycle exactly generates.
Relevant to distinguish between finite and infinite $\metaQmin$
is not the constant factor $k$, but only the number of iterations~$n$.}
Like that, $\gram$ causes $\metaQmin$ to be infinite by deriving $G_n$
with a path of length $n$ as a subgraph via: 
\begin{itemize}
\item 
\emph{directed} rule cycles and therefore fully balanced binary trees as subgraphs%
~(Lemma~\ref{theorem:tree2id})
\item 
\emph{undirected} rule cycles and therefore proper DAGs as shown in Lemma~\ref{lemma:chord2id}.
\end{itemize}
Consequently, $\gram$ generates a DAG $G_n$ whose size and number of meta-states 
depend on the number of rule cycle iterations which are unbounded. Like that, $\metaqmin$ is infinite.


The direction vice versa assumes an infinite $\metaqmin$ and arguments by negating the conditions one by one:
The existence of the stated rule cycle and rule path, as well as the requirements 
\lipicsfont{$c$}, \lipicsfont{$\Pi$} and \lipicsfont{$c\Pi$}
have to meet are iteratively shown to hold by negating them
and concluding a contradiction.
\begin{itemize}
\item Assume no $\Pi$ exists. By their definition, grammars without rule paths do not generate DAGs.
      With $\metaqmin$ being infinite, a rule path must exist. 
      Assume no $c$ exists.
      Languages without a rule cycle are finite (Lemma~\ref{theorem:linfiffrc}).
      Finite languages have a finite $\metaQmin$ (Lemma~\ref{theorem:finite}), contradicting our assumption that no rule cycle exists.
      Infinite languages are necessarily generated by cycles (Lemma~\ref{theorem:linfiffrc}).
      Thus, both rule path and cycle, called $\Pi$ and $c$, exist. 
\item 
      Assume condition \lipicsfont{$c$} does not hold.
      By definition, no rule cycles without marked nonterminals exist,
      thus $\m p$ must occur in $c$.
      A rule cycle without an unmarked nonterminal, here $q$,
      yields finitely many meta-states (Lemma~\ref{theorem:both}, 3.).
      Condition \lipicsfont{$c$} holds.
\item 
      Without condition 
      \lipicsfont{$\Pi$},
      thus without $\Pi$ being connected with its both ends $q$ and $p$ to $c$,
      there would be no chord path%
      \footnote{Recall the introduction of the proof for the explanation of $\Pi$ being a chord path.}.
      First, assume that there would be no rule path with $q$ as one of its ends.
      However, it is immediate that some rule path is indeed connected to $c$:
      Since $q$ occurs unmarked in $c$, a rule path $\Pi$ connected by $q$ necessarily exists,
      yielding -- without the chord -- an infinite
      set of meta-states, indeed. But, it is $\metaqzero$ which is infinite, not
      $\metaqmin$
      -- by Lemma~\ref{theorem:both} again
      -- in case $\Pi$ does not lead back to $c$ via $\m p$, our second assumption when negating
      \lipicsfont{$\Pi$}. This contradiction shows that a path $\Pi$ between $q$ and $p$ must exist for an infinite $\metaQmin$.
\item 
      Negating condition~\lipicsfont{$c\Pi$} means allowing
      $q$ as unmarked in a head of $c$ while it is in a tail of $\Pi$'s two weak rules or vice versa.
      This would induce the wrong orientation of $q$ so that $\Pi$ would not be a chord.
      And, we already know from
      the previous point 
      that $\Pi$ has to be a chord.
\end{itemize}
We cannot sacrifice any condition without contradicting our assumption of an infinite $\metaqmin$
which proves the second direction and by that completes the proof.
\end{proof}

What if we restrict a grammar $\gram$ with a set of meta-states $\metaQ$
instead of deriving this set out of the grammar?
    In that case the grammar's language possibly changes.
    Whereas when just extracting $\metaqmin$,
    the grammar's language is not altered.

\begin{definition}[\tbd{Finite meta-state DAG language (FD)}]\label{def:fdl}
A minimal deterministic grammar $\gram = (N, \Sigma, R)$
generates a \emph{finite meta-state language} $\LQ$
where a finite set $\metaQ \subseteq \mathbb{N}^Q$ is given to restrict which rules in $R$ can be used.
The derivation step $G_1 \To G_2$ is only allowed if the meta-state $\meta G_2 \in \metaQ$.
\end{definition}

Languages in the classes $\FID$ and $\FD$ can use their finite sets of meta-states,
$\metaQmin$ and $\metaQ$, respectively,
to construct a classical finite state automaton to recognize themselves.

\section{Classical Finite State Automata for DAG Languages}\label{sec:fsa}

This section describes how classical finite state automata come into play when recognizing
certain DAG languages.
An $\FD$ language can be recognized
by a classical deterministic finite state automaton (DFA).
As an example, see Figure~\ref{pic:astar} which shows the automaton for accepting DAGs of $L(\Gstar)$
as defined in Example~\ref{ex:star}.
The top-down reading process induces merely a partial order on the vertices.
The deterministic automaton thus reads a DAG in a partly nondeterministic fashion.

A DFA is a five-tuple $A = (Q, \Gamma, \delta, q_0, q_{\$})$
with the finite sets $Q$ and $\Gamma$ being the states and the alphabet, resp.,
$q_0, q_{\$} \in Q$ being the start and final state, resp., and
$\delta: (Q\times\Gamma)\to Q$ being the transition function,
extended inductively to strings $\delta (Q\times\Gamma^*)\to Q$ by applying $\delta$ symbol-wise.
We omit entries of $\delta$ if they do not lead to an accept state and thus consider only partial DFAs.

Every DAG grammar $\gramdef$ with $L(\gram)$ being a finite meta-state DAG language
gives rise to a DFA $M = (R, \metaQmin(\gram), \delta, \emptyset, \emptyset)$
such that for all $\metaq \in \metaQmin$ the transition function yields the following if  $\alpha \subseteq \metaq$ then
$$\delta(\metaq, (\alpha \sigmaarrow \beta)) = ( \metaq \setminus \alpha ) \cup \beta.$$ 

Although $M$ meets the requirements of a DFA,
we call its states in $\metaQmin$ meta-states instead to avoid confusion.
%
While reading a string with a DFA is as easy as reading it symbol by symbol,
reading a DAG is somewhat more complicated
since the vertices do not exhibit an obvious total order.
Therefore, $M$ reads the vertices partly nondeterministically.
Let $\dagdef$ be a DAG in $\D_\Sigma$.
Such a DAG has no labels assigned to its edges.
We denote this by $\lab(e) = \#$.
Since we are restricted to top-down procedures,
reading a vertex means assigning states to the outgoing edges:
Consequently, $M$ may read a vertex only if
all of its ingoing vertices have been assigned labels,
thus $\# \notin [\lab(\IN)]$.
If $\lab(v)$ matches the $\sigma$ of a rule $r = (\alpha \sigmaarrow \beta)$
such that $\delta(\metaq, r)$ yields a new state $\metaq'$ of $M$,
$M$ can read $v$ and assigns the labels $\beta$ to its outgoing edges:
$\lab(\OUT(v)) = \beta$.%
\footnote{Note that, by definition,
the empty DAG $\emptygraph$ is not in the language although
the automaton accepts it due to the start state as accepting state.
Only DAGs $G \in D_\Sigma$ are considered and $\emptygraph$ is not in~$D_\Sigma$.}

By reading the vertices top-down, $M$ will accept a DAG if it nondeterministically
chooses the right order of vertices.
The order in which the vertices are read is already restricted by imposing the requirement to read top-down.
But, we can improve upon this by fine-tuning the order in which the roots are read.
When $M$ cannot read any non-root, then there exists no vertex with all its ingoing edges labeled.
Instead of choosing an arbitrary root next 
-- since in a top-down a root has no prerequisites and can always be read top-down --
$M$ can choose a root which is needed next.
Which root $\vroot \in V$ do we need next? One where the DAG $G$
has a path from $\vroot$ to a vertex $\vnext$ whose ingoing vertices are labeled as well as not labeled.

We defined the language class $\FD$ at the end of the last section (c.f.~Definition~\ref{def:fdl}).
Now that we know how a grammar $\gram$ in combination with a finite set of meta-states $\metaQ(\gram)$
can determine a DFA for the DAG language $\Lgram$,
let us come back to the idea of restricting a grammar with a set of meta-states not derived from $\gram\!\!$.
Since languages easily become $\ID$ languages, the rules which are allowed to remain in the class of $\FID$
limit the expressiveness of languages.
The motivation behind restricting via a finite $\metaQ$
strives to increase the expressiveness of a language while keeping the number of meta-states
to recognize or generate it finite in order to profit from transferring the good algorithmic properties
of regular string languages. We use the language which comprises DAGs looking like a rainbow to illustrate
restricting a given $\ID$ language to become an $\FD$ language recognizable by an FSA, see Figure~\ref{fig:rainbow}.
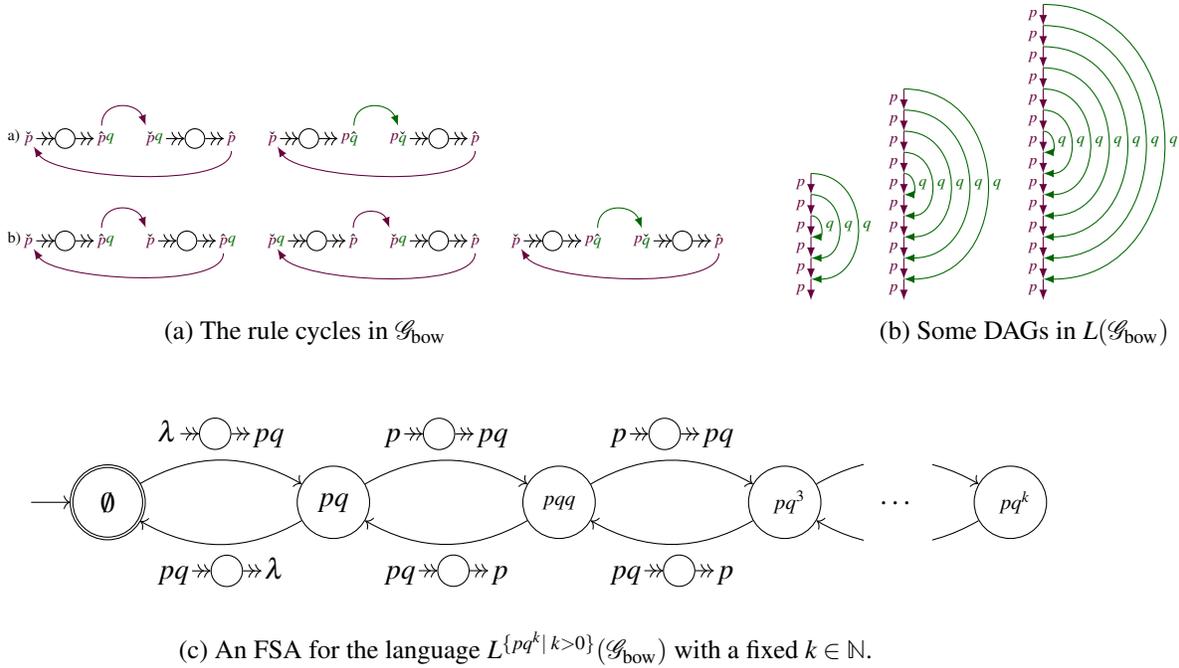
\begin{figure}[t]
\captionsetup[subfigure]{justification=centering}
\subfloat[The rule cycles in $\Gbow$]{
  \label{pic:bowrulecycle}
  \begin{minipage}{.5\textwidth}
  \tiny
    a)
    $
    \anchor{r1h0}{\men\p}
    \earrow
    \anchor{r1t0}{\mex\p}
    \anchor{r1t1}{\q}
\qquad
    \anchor{r2h0}{\men\p}
    \anchor{r2h1}{\q}
    \earrow
    \anchor{r2t0}{\mex\p}
\qquad
    \anchor{R1h0}{\men\p}
    \earrow
    \anchor{R1t0}{\p}
    \anchor{R1t1}{\mex\q}
\qquad
    \anchor{R2h0}{\p}
    \anchor{R2h1}{\men\q}
    \earrow
    \anchor{R2t0}{\mex\p}
    $
    \begin{tikzpicture}[remember picture,overlay]
            \path (r1t0) edge[-latex,out=90,in=110,looseness=1.5,\pcolor] (r2h0)
                  (r2t0) edge[-latex,out=270,in=302,looseness=0.5,\pcolor] (r1h0);
            \path (R1t1) edge[-latex,out=90,in=110,looseness=1.5,\qcolor] (R2h1)
                  (R2t0) edge[-latex,out=270,in=302,looseness=0.5,\pcolor] (R1h0);
    \end{tikzpicture}

\vspace{1cm}

b)
    $
    \anchor{r1h0}{\men\p}
    \earrow
    \anchor{r1t0}{\mex\p}
    \anchor{r1t1}{\q}
\qquad
    \anchor{r2h0}{\men\p}
    \earrow
    \anchor{r2t0}{\mex\p}
    \anchor{r2t1}{\q}
\qquad
    \anchor{R1h0}{\men\p}
    \anchor{R1h1}{\q}
    \earrow
    \anchor{R1t0}{\mex\p}
\qquad
    \anchor{R2h0}{\men\p}
    \anchor{R2h1}{\q}
    \earrow
    \anchor{R2t0}{\mex\p}
\qquad
    \anchor{R1H0}{\men\p}
    \earrow
    \anchor{R1T0}{\p}
    \anchor{R1T1}{\mex\q}
\qquad
    \anchor{R2H0}{\p}
    \anchor{R2H1}{\men\q}
    \earrow
    \anchor{R2T0}{\mex\p}
    $
    \begin{tikzpicture}[remember picture,overlay]
            \path (r1t0) edge[-latex,out=90,in=110,looseness=1.5,\pcolor] (r2h0)
                  (r2t0) edge[-latex,out=270,in=302,looseness=0.5,\pcolor] (r1h0);
            \path (R1t0) edge[-latex,out=90,in=110,looseness=1.5,\pcolor] (R2h0)
                  (R2t0) edge[-latex,out=270,in=302,looseness=0.5,\pcolor] (R1h0);
            \path (R1T1) edge[-latex,out=90,in=110,looseness=1.5,\qcolor] (R2H1)
                  (R2T0) edge[-latex,out=270,in=302,looseness=0.5,\pcolor] (R1H0);
    \end{tikzpicture}
    \\~\\
  \end{minipage}
  ~\\
}
\hspace{2cm}
\subfloat[Some~DAGs~in~$L(\Gbow)$]{
  \label{pic:rainbow}
  \begin{minipage}{.4\textwidth}
\begin{tikzpicture}[node distance=8pt]
\tiny 

    \node [vertex] (xnext0) {};
\foreach \n [remember=\n as \lastn (initially 0)] in {1,...,3} {
    \node [vertex] (x\n)     [below=of xnext\lastn] {};
    \node [vertex] (xnext\n) [below=of x\n]     {};
    \path [-latex,color=tyrianpurple]  (xnext\lastn) edge node[left] {$p$} (x\n)
          (x\n) edge node[left] {$p$}  (xnext\n) ;
}
    \draw [-latex,color=pakistangreen, bend left,looseness=1.5,out=90,in=90]
          (xnext0) edge node[right] {$q$} (x3)
          (x1) edge node[right] {$q$}     (xnext2)
          (xnext1) edge node[right] {$q$} (x2)
	  ;

\end{tikzpicture}
\begin{tikzpicture}[node distance=8pt]
\tiny 

    \node [vertex] (xnext0) {};
\foreach \n [remember=\n as \lastn (initially 0)] in {1,...,5} {
    \node [vertex] (x\n)     [below=of xnext\lastn] {};
    \node [vertex] (xnext\n) [below=of x\n]     {};
    \path [-latex,color=tyrianpurple]  (xnext\lastn) edge node[left] {$p$} (x\n)
          (x\n) edge node[left] {$p$}  (xnext\n) ;
}
    \draw [-latex,color=pakistangreen, bend left,looseness=1.5,out=90,in=90]
          (xnext0) edge node[right] {$q$} (x5)
          (x1) edge node[right] {$q$}     (xnext4)
          (xnext1) edge node[right] {$q$} (x4)
          (x2) edge node[right] {$q$}     (xnext3)
          (xnext2) edge node[right] {$q$} (x3)
	  ;

\end{tikzpicture}
\begin{tikzpicture}[node distance=8pt]
\tiny 

    \node [vertex] (xnext0) {};
\foreach \n [remember=\n as \lastn (initially 0)] in {1,...,7} {
    \node [vertex] (x\n)     [below=of xnext\lastn] {};
    \node [vertex] (xnext\n) [below=of x\n]     {};
    \path [-latex,color=tyrianpurple]  (xnext\lastn) edge node[left] {$p$} (x\n)
          (x\n) edge node[left] {$p$}  (xnext\n) ;
}
    \draw [-latex,color=pakistangreen, bend left,looseness=1.5,out=90,in=90]
          (xnext0) edge node[right] {$q$} (x7)
          (x1) edge node[right] {$q$}     (xnext6)
          (xnext1) edge node[right] {$q$} (x6)
          (x2) edge node[right] {$q$}     (xnext5)
          (xnext2) edge node[right] {$q$} (x5)
          (x3) edge node[right] {$q$}     (xnext4)
          (xnext3) edge node[right] {$q$} (x4)
	  ;

\end{tikzpicture}
\end{minipage}
}
\\~\\~\\
\subfloat[An FSA for the language $L^{\{ pq^k |\; k>0 \}}(\Gbow)$ with a fixed $k \in \mathbb N$.]{
  \label{pic:abow}

  \begin{tikzpicture}[->,node distance=3cm,
    ever state/.style={inner sep=4cm},
    initial text=]
\node[initial,state,accepting] (0) {$\emptyset$};
\node[state] (r) [right of=0] {$pq$};
\node[font=\scriptsize,state, align=center] (rr) [right of=r,text width=5mm] {$pqq$~};
\node[font=\scriptsize,state, align=center] (rrr) [right of=rr,text width=3mm] {$pq^3$};
\node[font=\scriptsize,state, align=center] (rrrr) [right of=rrr,text width=3mm] {$pq^k$};

\path
( 0) edge [bend left]  node   [above]      { \bowrzero}         ( r)
( r) edge [bend left]  node   [above]      { \bowrone}         (rr)
(rr) edge [bend left]  node   [above]      { \bowrone}         (rrr)
(rrr) edge [bend left]  node   [above]      {}         (rrrr)
(rrrr) edge [bend left]  node   [below]      {}         (rrr)
(rrr) edge [bend left]  node   [below]      { \bowrtwo}         ( rr)
(rr) edge [bend left]  node   [below]      { \bowrtwo}         ( r)
( r) edge [bend left]  node   [below]      { \bowrtwozero}         ( 0)
%
%
%
;

\path (rrr) -- (rrrr) node[midway,shape=rectangle,fill=white, minimum height=3cm] (dots)
{~\dots ~~};
\end{tikzpicture}
}
\captionsetup{subrefformat=parens}
\caption{The grammar $\Gbow$ generates, among others, DAGs looking like rainbows (\emph{vertices implicit}) \subref{pic:bowrulecycle} via the rule cycles in \subref{pic:rainbow}. For a fixed $k \in \mathbb N$, restricting the grammar by $\metaQ = \{pq^k |\; k>0 \}$ allows the construction of an FSA and imposes a bound on the number of bows in a `rainbow'.}\label{fig:rainbow}
\end{figure}

\begin{theorem}[\tbd{FD $\nsubseteq$ $\RDLdet$}]\label{theorem:metastatebounded}
The class of finite meta-state DAG languages is not a subset
of the class of top-down deterministic regular DAG languages $\RDLdet\!$.
\end{theorem}
\begin{proof}
Given is a minimal deterministic DAG grammar $\gram = (N, \Sigma, R)$ and a finite set of meta-states $\metaQ$.
If FD $\nsubseteq$ $\RDLdet$ holds, then
there exists a DAG language $\LQ$ that is not in the class of $\RDLdet\!$.
Our $\gram$ is a deterministic DAG grammar and as such generates a language $\Lgram \in \RDLdet\!$. 

We try to construct a language not in $\RDLdet\!$ by limiting $\gram$'s language to a finite one.
Suppose, $\gram$ generates an infinite language $\Lgram$
and $\metaQ$ prevents derivations of any rule cycle in $R$.
According to Lemma~\ref{theorem:linfiffrc}, then, $\LQ$ will be finite.
However, Lemma~\ref{theorem:finite} tells us that
any finite language can be generated by a deterministic DAG grammar $\gram' \neq \gram$
such that~$L(\gram') = \LQ$. This attempt did not work out.

Our next attempt is the grammar $\Gbow$ with its infinite language $\Lbow \in \ID$.
Its rule set comprises $
(\lambda \rarrow{r} pq      ), 
(p       \rarrow{o} pq      ), 
(pq      \rarrow{c} p       )%
$ and 
$%
(pq      \rarrow{l} \lambda )$.
The grammar $\Gbow$ can among others generate DAGs similar to garlands and rainbows.
To generate an arbitrary long DAG looking like a garland,
the set of meta-states
$\metaQ = \{\emptygraph, p, pq\}$ suffices.
If, however, $\Gbow$ generates a DAG looking like a rainbow,
it repeats the rule with the vertex label~$\osymbol$
an unbounded number of times.
\footnote{%
Allowing pennants spanning more than one vertex is possible, too.
With $\metaQ = \{pq^n \mid p,q \in N \text{ and } n \in \mathbb{N}\}$
the bow(s) for the pennants in a DAG $G_{garland} \in \LQ$ can span the maximum of $n$ vertices.%
}
Such a rainbow DAG in $\LQ$, with the length of the bow(s) limited,
can only be generated by a deterministic DAG grammar if the grammar's language is finite.
But, a finite language cannot include DAGs with $n$-sized bows of arbitrary length.
On the other hand, any grammar generating garland DAGs of unbounded length will
also include the rainbow DAGs due to the possibility to swap the edges (Lemma~\ref{theorem:swap}).
And, according to Lemma~\ref{theorem:unbound}, we can swap the edges
since the edge labels (the nonterminals) must be repeated for a DAG of unbounded length
and distinct edges labeled with the same nonterminal in a derivation DAG can be swapped.
Thus, restricting the language to $\LQbow$ by above given $\metaQ$
results in a language not in $\RDLdet\!$,
completing the proof.
\end{proof}

Again, swapping (Lemma~\ref{theorem:swap}) would allow us
to generate unbounded rainbows, so we conclude:
restricting a grammar by a set of meta-states
can prevent the swapping operation.

\begin{corollary}[Swapping in $\FD$]\label{cor:FDnoswap}
For a grammar $\gram$ with its $\LQ \in \FD\:\backslash\:\FID$ holds:
\begin{itemize}
\item $\exists\:\gd\!\in\!\LQ: \gswap \notin \LQ$
\item $| \LQ | < | \Lgram |$
\item $\Lgram = \{\:\gswap\:|\:\swapdee \in \LQ\:\}$
\item $\Lgram \in \ID$
\end{itemize}
\end{corollary}

Contrary to $\FID$-grammars on whose derivation DAGs swapping is always allowed,
the derivation DAGs of grammars generating languages in $\LQ \in \FD \backslash \FID $
the swapping operation is restricted by the given set of meta-states $\metaQ$.
Restricting a grammar by allowing only certain meta-states
corresponds to restricting swapping on the derivation DAGs.
Via those restrictions, DAGs are lost which cannot be accepted without the missing meta-states
resulting in a language with less graphs.
The transitive closure of the swapping operation on the language $\LQ$
returns the original language $\Lgram$ which must be in the class $\ID$.

%
%
%

For free -- by $\FD$ definition -- we can observe the closure properties valid for DFAs,
since $\FD$ languages are recognized by DFAs.
\begin{observation}[$\FD$ -- Closure under Union and Intersection]\label{obs:TODODEFLABEL}
~\\The language class $\FD$ is closed under union and intersection.
\end{observation}

\section{Conclusion}

We have defined the DAG language classes $\ID$ and $\FD$
and characterized them.
By imposing the set of meta-states as a given restriction,
we additionally have defined $\FD$, which intersects with $\RDL$ but is not a subset of it -- adding to expressiveness. 
For languages in $\FID$ and $\FD$,
we proved that it is possible to construct a classical string automaton recognizing the language.
Analysing $\ID$ further, we expect similarities with the Chomsky hierarchy.

\begin{pdconjecture*}\label{conj:id}%
Analogous to languages in $\FD$ being recognized by a finite state automaton,
we conjecture all languages in $\ID$ to be recognized by a pushdown automaton.
\end{pdconjecture*}

The FSA construction by meta-states,
could be applied not only to the top-down deterministic version, but also to the plain regular DAG automaton.
By dropping the determinism restriction, similar to dropping
the planarity restriction \cite{DBLP:journals/corr/abs-1810-12266} imposed
in \cite{kamimura-slutzki:81}, possibly the NP-completeness for the membership problem could be tackled.
Imposing useful restrictions to be provided by the set of meta-states~$\metaQ$,
as the language class $\FD$ requires it,
is the task for more application centric research,
like semantic NLP parsing~\cite{quernheim-knight:12b,banarescu.etal:13}.
\newpage

\bibliographystyle{eptcs}
\bibliography{bibliography}

\newpage
\section{Appendix}\label{sec:appendix}
\printProofs
\end{document}